\documentclass[11pt]{article}
\usepackage{amsmath,amssymb,amsthm} 
\usepackage{bbm}
\usepackage{stmaryrd}
\usepackage{paralist}
\usepackage{hyperref}
\usepackage{color}
\usepackage{multirow}

\setdefaultenum{(i)}{}{}{}
\usepackage[margin=1in]{geometry}
\usepackage[pdftex]{graphicx}
\usepackage{subfigure}
\theoremstyle{plain}
\newtheorem{mythm}{Theorem} \numberwithin{mythm}{section}

\newtheorem{mylemma}[mythm]{Lemma}
\newtheorem{mydef}[mythm]{Definition}
\newtheorem{myrek}[mythm]{Remark}

\newtheorem{mycor}[mythm]{Corollary}

\DeclareMathAlphabet\scr{U}{scr}{m}{n}
\SetMathAlphabet\scr{bold}{U}{scr}{b}{n}
  \DeclareFontFamily{U}{scr}{\skewchar\font'177}%
  \DeclareFontShape{U}{scr}{m}{n}{<-6>rsfs5<6-8>rsfs7<8->rsfs10}{}%
  \DeclareFontShape{U}{scr}{b}{n}{<-6>rsfs5<6-8>rsfs7<8->rsfs10}{}%

\numberwithin{equation}{section}

\renewenvironment{thebibliography}[1]{%
\begin{oldthebibliography}{#1}%
\setlength{\baselineskip}{.9em}
\linespread{.9}
\small
\setlength{\parskip}{0ex}%
\setlength{\itemsep}{.1em}%
}%
{%
\end{oldthebibliography}%
}

\DeclareMathOperator{\lip}{LiPr}
\DeclareMathOperator{\sht}{ShTu}
\DeclareMathOperator{\wet}{WeTu}
\DeclareMathOperator{\esr}{ESR}

\begin{document}
\title{\vspace{-1.25cm} Portfolio Selection with Small Transaction Costs\\ and Binding Portfolio Constraints\footnote{The authors thank Christoph Czichowsky, Michael Grill, Paolo Guasoni, and Mete Soner for fruitful discussions. They are also grateful to two anonymous referees for their pertinent remarks.}}
\author{
Ren Liu
\thanks{ETH Z\"urich, Departement Mathematik, R\"amistrasse 101, CH-8092, Z\"urich, Switzerland, email \texttt{ren.liu@math.ethz.ch}}
\and 
Johannes Muhle-Karbe 
\thanks{ETH Z\"urich, Departement Mathematik, R\"amistrasse 101, CH-8092, Z\"urich, Switzerland, and Swiss Finance Institute, Walchestrasse 9, CH-8006, Z\"urich, Switzerland, email \texttt{johannes.muhle-karbe@math.ethz.ch}. Partially supported by the National Centre of Competence in Research ``Financial Valuation and Risk Management'' (NCCR FINRISK), Project D1 (Mathematical Methods in Financial Risk Management), of the Swiss National Science Foundation (SNF).} 
}
\date{}
\pagestyle{plain}
\maketitle

\begin{abstract}
An investor with constant relative risk aversion and an infinite planning horizon trades a risky and a safe asset with constant investment opportunities, in the presence of small transaction costs and a binding exogenous portfolio constraint. We explicitly derive the optimal trading policy, its welfare, and implied trading volume. As an application, we study the problem of selecting a prime broker among alternatives with different lending rates and margin requirements. Moreover, we discuss how changing regulatory constraints affect the deposit rates offered for illiquid loans.
\end{abstract}




\section{Introduction}

\emph{Transaction costs} and \emph{trading constraints} are two central frictions in financial markets, but most of the literature focuses either on transaction costs \cite{magill.constantinidis.76,constantinides.86,davis.norman.90,dumas.luciano.91,shreve.soner.94} or on portfolio constraints \cite{vila.zariphopoulou.90,grossman.vila.92,cvitanic.karatzas.93} -- separately -- with the exception of the paper by Dai, Jin, and Liu \cite{dai.al.11}. They describe the value function of a finite-horizon model with both frictions by means of a double-obstacle problem, obtain some monotonicity properties of the optimal trading boundaries, and conduct an extensive numerical analysis.

Our goal is to obtain more tractable results. To this end, we work with the infinite-horizon model of Dumas and Luciano \cite{dumas.luciano.91}: an investor with constant relative risk aversion trades to maximize the long-term growth rate of her utility, in a market with one risky asset following geometric Brownian motion and one safe asset with constant interest rate. In the presence of a binding exogenous portfolio constraint, i.e., an upper bound on the risky weight, and proportional transaction costs, we derive the optimal trading strategy as well as the associated welfare and trading volume. As in the unconstrained case \cite{gerhold.al.11}, all formulas are explicit in terms of the market and preference parameters as well as an additional \emph{transaction cost gap}, characterized as the root of a scalar equation. As the spread becomes small, all quantities admit explicit asymptotic expansions, in terms of model and preference parameters only. 

These results help to clarify the joint impact of transaction costs and constraints. 

Constantinides \cite{constantinides.86} observed that -- in the absence of constraints -- ``transaction costs have a first order effect on asset demand". On the other hand, he found that their welfare impact is typically less pronounced, since ``a small liquidity premium is sufficient to compensate an investor for deviating significantly from the target portfolio proportions". Starting with Shreve and Soner~\cite{shreve.soner.94}, these numerical results have been made precise in an asymptotic manner: deviations from the optimal frictionless portfolio are of order $\epsilon^{1/3}$ as the spread $\epsilon$ becomes small \cite{janecek.shreve.04}, but the corresponding liquidity premium and utility loss are only of order $\epsilon^{2/3}$ \cite{shreve.soner.94}. These results are robust within the class of diffusion models \cite{soner.touzi.11}, but break down in the presence of binding constraints. Here, we find that the effects on portfolio composition and welfare are of the same order $\epsilon^{1/2}$, in line with the numerical observation of \cite{dai.al.11}  that ``transaction costs can have a first-order effect'' on welfare. 

In the absence of transaction costs, the impact of binding portfolio constraints is easily quantified: harder constraints simply reduce the investor's risky weight. Considering them jointly with trading costs additionally allows to asses their impact on the share turnover the investor's rebalancing generates. The latter turns out to be of order $\epsilon^{-1/2}$, and therefore dominates its frictionless counterpart of order $\epsilon^{-1/3}$ for sufficiently small spreads. The corresponding comparative statics strongly depend on whether the investor's position is leveraged or not. In the absence of leverage, tighter constraints increase trading volume, and constrained turnover dominates its unconstrained counterpart even for large transaction costs. On the contrary, if the optimal policy prescribes a leveraged position in the risky asset, then sufficiently tight constraints reduce turnover, and unconstrained turnover need only be dominated for very small spreads.

Even though all individual asymptotic rates change compared to the unconstrained case, one key observation from the unconstrained case \cite{gerhold.al.11} remains valid: The welfare impact of transaction costs equals implied trading volume (measured consistently) times the spread, times a constant.

To illustrate the implications of our results, we discuss two applications. First, we consider an investor choosing which prime broker to use to buy a leveraged risky position \emph{on margin}. Each broker is prepared to let the investor borrow at a specific rate and up to a given leverage constraint, and our results allow to quantify the attractiveness of each combination. We find that among those brokers that are equally attractive in the absence of transaction costs, the investor typically prefers those with harder constraints (i.e., high margin requirements) and lower lending rates, because transaction costs make highly leveraged portfolios less attractive than they appear to be in frictionless markets. As a second application, we consider a bank that can borrow from its depositors at the safe rate to provide long-term loans, whose book values are assumed to follow geometric Brownian motion. Then, the bank's optimization problem is precisely of the type considered above, and the portfolio constraints correspond to the minimum capital requirements imposed by regulatory authorities. This model can in turn be used to assess the impact of tighter regulatory constraints: Assuming that the bank aims to achieve the same performance, it will have to decrease its deposit rate to compensate for the negative welfare effect of harder constraints. The size of this effect crucially depends on the liquidity of the long-term loans the bank is providing, i.e., on the transaction costs incurred when prematurely liquidating them. Since leveraged positions are less attractive with transaction costs, the decrease in the banks deposit rate is most pronounced in the perfectly liquid case, and is diminished substantially if illiquidity is accounted for. 

On a mathematical level, the proof of our verification theorem is based on applying a variant of the argument of Guasoni and Robertson \cite{guasoni.robertson.12} to a fictitious shadow price trading without transaction costs, which is equivalent to the original market with transaction costs both in terms of the optimal strategy and utility. In contrast to the unconstrained case \cite{gerhold.al.11}, the dynamics of the shadow price involve a singular component, such that the corresponding frictionless market is not arbitrage-free. The constraints,  however, prevent the optimal strategy from exploiting these arbitrage opportunities.

 The remainder of the paper is organized as follows: Section 2 describes the setting and states the main results, whose implications are discussed in Section 3. The main results are derived heuristically in Section 4, and proved rigorously in Section 5.
       

\section{Model and Main Result}
Consider a market with one safe asset $S^0_t = e^{rt}$, $r>0$, and one risky asset, whose \emph{ask (buying) price} $S_t$ follows geometric Brownian motion:
\begin{displaymath}
d S_t/S_t = (\mu + r)d t + \sigma d W_t,\quad S_0 \in (0,\infty).
\end{displaymath}
Here, $\mu >0$ is the expected excess return, $\sigma > 0$ the volatility, and $(W_t)_{t \geq 0}$ is a Brownian motion. The corresponding \emph{bid (selling) price} is $(1-\epsilon)S_t$, where $\epsilon \in (0,1)$ represents the width of the relative bid-ask spread.\footnote{This notation is equivalent to the usual setup with the same constant proportional transaction costs for purchases and sales \cite{davis.norman.90, janecek.shreve.04, shreve.soner.94}. Indeed, set $\check{S}_t=\frac{2-\epsilon}{2}S_t$ and $\check{\epsilon}=\frac{\epsilon}{2-\epsilon}$. Then $((1-\epsilon)S_t,S_t)$ coincides with $((1-\check{\epsilon})\check{S}_t,(1+\check{\epsilon})\check{S}_t)$. Conversely, any bid-ask process $((1-\check{\epsilon})\check{S}_t,(1+\check{\epsilon})\check{S}_t)$ with $\check{\epsilon} \in (0,1)$ equals $((1-\epsilon)S_t,S_t)$ for $S_t=(1+\check{\epsilon})\check{S}_t$ and $\epsilon=\frac{2\check{\epsilon}}{1+\check{\epsilon}}$.}

 A self-financing \emph{trading strategy} is an $\mathbb{R}^2$-valued, predictable process  $(\phi^0,\phi)$ of finite variation: $(\phi_{0-}^0, \phi_{0-})= (\xi^0,\xi)\in \mathbb{R}^2$ denotes the initial positions (in units) in the safe and risky asset, and $(\phi_t^0,\phi_t)$ denotes the positions held at time $t$. Writing $\phi_t=\phi_t^{\uparrow}-\phi_t^{\downarrow}$ as the difference between the cumulative number of shares bought ($\phi_t^{\uparrow}$) and sold ($\phi_t^{\downarrow}$) by time $t$, the \emph{self-financing condition} states that the safe position only changes due to trades in the risky asset:
\begin{equation}\label{eq:sf}
S_t^0 d\phi^0_t = -S_t d\phi_t^{\uparrow}+(1-\epsilon)S_t d\phi_t^{\downarrow}, \quad \forall t \geq 0.
\end{equation}

As is customary, attention is restricted to strategies that remain solvent at all times. In addition, we focus on strategies that satisfy an exogenous \emph{portfolio constraint}, i.e., whose risky weight is uniformly bounded from above.

 \begin{mydef}
A self-financing trading strategy $(\phi^0,\phi)$ is called \emph{admissible}, if its liquidation value is positive at all times,
$$\Xi_t^{\phi} := \phi_t^0S_t^0+(1-\epsilon)S_t\phi_t^{+}-S_t\phi_t^{-} \geq 0, \quad  \forall t \geq 0,$$
and it satisfies the \emph{portfolio constraint}
\begin{equation}\label{leverage constraint}
\pi_t:=\frac{\phi_t S_t}{\phi^0_t S^0_t+\phi_t S_t} \leq \pi_{\max}, \quad  \forall t \geq 0.
\end{equation}
\end{mydef}

It will sometimes be notationally convenient to decompose the constraint $\pi_{\max}$ in a multiplicative manner as
$$\pi_{\max}=\kappa \pi_*,$$
where 
$$\pi_*=\mu/\gamma\sigma^2$$
is the frictionless optimal Merton proportion and $\kappa$ is the \emph{relative constraint}. Throughout, the constraints are assumed to be binding, i.e., $\pi_{\max}<\pi_*$ resp.\ $\kappa<1$. Otherwise, they have no effect as in the frictionless case if the transaction costs $\epsilon$ are sufficiently small.

As in Dumas and Luciano~\cite{dumas.91} the investor has constant relative risk aversion $0<\gamma \not= 1$ and an infinite planning horizon: 
\begin{mydef}\label{longrun}
An admissible strategy $(\varphi^0,\varphi)$ is called \emph{long-run optimal}, if it maximizes the equi\-valent safe rate
\begin{equation}\label{defdef}
\liminf_{T\rightarrow \infty}\frac{1}{T}\log{\mathbb{E}\left[(\Xi_T^{\varphi})^{1-\gamma}\right]^{\frac{1}{1-\gamma}}}
\end{equation}
over all admissible strategies, where $0<\gamma\not=1$ denotes the investor's relative risk aversion.
\end{mydef}

Our main results can be summarized as follows:

\begin{mythm}\label{main result}
An investor with constant relative risk aversion $0 < \gamma \not=1$ trades to maximize the equivalent safe rate in the presence of a binding portfolio constraint $0 < \pi_{\max}=\kappa \pi_* \not=1$.\footnote{In the degenerate case $\pi_{\max} = \kappa \pi_* =1$, an optimal strategy is to fully invest into the risky asset initially and never trade again afterwards. Both share and wealth turnover vanish in this case, and the equivalent safe rate and liquidity premium coincide with their counterparts in the absence of transaction costs.} 

Then, for small transaction costs $\epsilon > 0$:
\begin{enumerate}
\item[i)] \emph{(Equivalent Safe Rate)}\\
For the investor, trading the risky asset with transaction costs and portfolio constraints is equivalent to leaving all wealth in a hypothetical safe asset, which pays the higher \emph{equivalent safe rate}
\begin{displaymath}
\esr = r + \frac{\mu^2}{2\gamma\sigma^2}(2\kappa/(1-\lambda)-\kappa^2)(1-\lambda)^2,
\end{displaymath}
where the \emph{gap} $\lambda$ is defined in Item~iv) below.
\item[ii)] \emph{(Liquidity Premium)}\\
Trading the risky asset with transaction costs and constraints is equivalent to trading a hypothetical asset with no transaction cost and no constraint, with the same volatility $\sigma$, but with lower expected excess return $\mu\sqrt{2\kappa/(1-\lambda)-\kappa^2}(1-\lambda)$. Thus, the \emph{liquidity premium} is
\begin{displaymath}
\lip = \mu - \mu\sqrt{2\kappa/(1-\lambda)-\kappa^2}(1-\lambda).
\end{displaymath}
\item[iii)]\label{opt} \emph{(Trading Policy)}\\
It is optimal to keep the risky weight (in terms of the ask price) between the buying and selling boundaries
\begin{displaymath}
\pi_- = (1-\lambda)\pi_{\max}, \quad \pi_+ = \pi_{\max}.
\end{displaymath}
\item[iv)]\label{Gap} \emph{(Gap)}\\
$\lambda$ is the unique value for which the solution of the initial value problem
\begin{eqnarray}
0 &=& w'(x) + (1-\gamma)w(x)^2 + (\tfrac{2\mu}{\sigma^2}-1)w(x) -\tfrac{\mu^2}{\gamma\sigma^4}(1-(1-\kappa(1-\lambda))^2),\nonumber\\
w(0) &=& (1-\lambda)\pi_{\max},\nonumber
\end{eqnarray}
also satisfies the terminal value condition
\begin{displaymath}
w\left(\log{\left(u/l(\lambda)\right)}\right)= \tfrac{\pi_{\max}(1-\epsilon)}{(1-\pi_{\max})+\pi_{\max}(1-\epsilon)}=:w_{+},
\end{displaymath}
where
\begin{displaymath}
u/l(\lambda)= \tfrac{\pi_{\max}/(1-\pi_{\max})}{(1-\lambda)\pi_{\max}/(1-(1-\lambda)\pi_{\max})}.
\end{displaymath}
\item[v)] \emph{(Share Turnover)}\\
Share Turnover, defined as shares traded $d\|\varphi\|_t = d\varphi_t^{\uparrow}+d\varphi_t^{\downarrow}$ divided by shares held $|\varphi_t|$, has the long-term average:
\begin{displaymath}
\sht  := \lim_{T \rightarrow \infty}\frac{1}{T}\int_0^T\frac{d\|\varphi\|_t}{|\varphi_t|}
 =  \begin{cases} \frac{\sigma^2}{2}\left(\frac{2\mu}{\sigma^2}-1\right)\left(\frac{1-\pi_-}{(u/l(\lambda))^{\frac{2\mu}{\sigma^2}-1}-1}-\frac{1-w_{+}}{(u/l(\lambda))^{1-\frac{2\mu}{\sigma^2}}-1}\right), & \text{if } \mu \not= \frac{\sigma^2}{2},\\
 \frac{\sigma^2}{2\log{(u/l(\lambda))}}\left(1-\pi_- + 1-w_{+}\right), & \text{if } \mu = \frac{\sigma^2}{2}.
\end{cases}
\end{displaymath}
\item[vi)] \emph{(Wealth Turnover)}\\
Wealth Turnover, defined as wealth traded divided by wealth held, has the long-term average:
\begin{align*}
\wet & :=  \lim_{T\rightarrow \infty}\frac{1}{T}\left(\int_0^T\frac{S_t d\varphi_t^\uparrow}{\varphi_t^0S_t^0+\varphi_t S_t} + \int_0^T \frac{(1-\epsilon)S_t d\varphi_t^\downarrow}{\varphi_t^0S_t^0+\varphi_t(1-\epsilon)S_t}\right)\nonumber \\  &=
\begin{cases}
\frac{\sigma^2}{2}\left(\frac{2\mu}{\sigma^2}-1\right)\left(\frac{\pi_-(1-\pi_-)}{(u/l(\lambda))^{\frac{2\mu}{\sigma^2}-1}-1}-\frac{w_+(1-w_+)}{(u/l(\lambda))^{1-\frac{2\mu}{\sigma^2}}-1}\right), & \text{if } \mu \not= \frac{\sigma^2}{2},\\
\frac{\sigma^2}{2\log{(u/l(\lambda))}}\left(\pi_-(1-\pi_-)+ w_+(1-w_+)\right), & \text{if } \mu = \frac{\sigma^2}{2}.
\end{cases}\nonumber
\end{align*}
\item[vii)] \emph{(Asymptotics)}\\
The following asymptotic expansions hold true as $\epsilon \downarrow 0$:\footnote{Notice that $\kappa$ and $\pi_{\max}/\pi_{*}$ are used interchangeably to ease the computations for the comparative statics.}
\begin{align*}
\esr &= 
r + \frac{\mu^2}{2\gamma\sigma^2}\left(\frac{2\pi_{\max}}{\pi_*}- \left(\frac{\pi_{\max}}{\pi_*}\right)^2\right)- \gamma\sigma^2\left(\frac{1}{\gamma}(\pi_*-\pi_{\max})(1-\pi_{\max})^2\pi_{\max}^2\right)^{1/2}\epsilon^{1/2}+ \mathcal{O}(\epsilon), \nonumber\\
\lip &= 
\mu\left(1-\sqrt{\frac{2\pi_{\max}}{\pi_*}-\left(\frac{\pi_{\max}}{\pi_*}\right)^2}\right)+  \gamma\sigma^2\left(\frac{1}{\gamma}\frac{\pi_{\max}(\pi_*-\pi_{\max})(1-\pi_{\max})^2}{2\pi_{*}-\pi_{\max}}\right)^{1/2}\epsilon^{1/2}+ \mathcal{O}(\epsilon), \nonumber\\
\pi_- &= 
\pi_{\max}- \left(\frac{1}{\gamma}\frac{\pi_{\max}^2}{\pi_*-\pi_{\max}}(1-\pi_{\max})^2\right)^{1/2}\epsilon^{1/2}+ \mathcal{O}(\epsilon),\nonumber\\
\lambda &= \left(\frac{1}{\gamma}\frac{(1-\pi_{\max})^2}{\pi_*-\pi_{\max}}\right)^{1/2}\epsilon^{1/2}+ \mathcal{O}(\epsilon),\nonumber\\
\sht &= \gamma\sigma^2\left(\frac{1}{\gamma}(\pi_*-\pi_{\max})(1-\pi_{\max})^2\right)^{1/2}\epsilon^{-1/2}+ \mathcal{O}(1),\nonumber\\
\wet &=  \gamma\sigma^2\left(\frac{1}{\gamma}(\pi_*-\pi_{\max})(1-\pi_{\max})^2\pi_{\max}^2\right)^{1/2}\epsilon^{-1/2}+ \mathcal{O}(1).\nonumber
\end{align*}
\end{enumerate}
\end{mythm} 

As in the absence of constraints, the stationary policy from Theorem \ref{main result} is also approximately optimal -- at the leading order $\epsilon^{1/2}$ for small transaction costs $\epsilon$ -- for any \emph{finite} time horizon:

\begin{mythm}\label{finitehorizone}
Fix a time horizon $T>0$. Then, the finite-horizon equivalent safe rate of any admissible strategy $(\psi^0,\psi)$ satisfies the upper bound
\begin{eqnarray}
\frac{1}{T}\log{\mathbb{E}[(\Xi_T^{\psi})^{1-\gamma}]}^{\frac{1}{1-\gamma}} &\leq& 
\esr +\frac{1}{T} \log{(\xi_{0^-}^0+\xi_{0^-}S_0)}+\mathcal{O}\left(\epsilon\right).\nonumber
\end{eqnarray}
The finite-horizon equivalent safe rate of the long-run optimal strategy $(\varphi^0,\varphi)$ from Theorem \ref{main result} satisfies the lower bound
\begin{eqnarray}
\frac{1}{T}\log{\mathbb{E}[(\Xi_T^{\varphi})^{1-\gamma}]}^{\frac{1}{1-\gamma}} \geq
\esr + \frac{1}{T} \log{(\xi_{0^-}^0+\xi_{0^-}S_0)} + \mathcal{O}\left(\epsilon\right).\nonumber
\end{eqnarray}
\end{mythm}

As pointed out by an anonymous referee, the finite horizon bounds in Theorem \ref{finitehorizone} also show that the policy from Theorem \ref{main result} maximizes \eqref{defdef} if the limes inferior is replaced by a limes superior.

\section{Implications}

Let us now discuss some of the implications of our main result.

\subsection{Asset Demand and Welfare}\label{order12}
In the absence of constraints, the impact of transaction costs on the optimal policy is large ($\sim \epsilon^{1/3}$, cf.\ \cite{janecek.shreve.04}) whereas their effect on welfare is small ($\sim \epsilon^{2/3}$, cf.\ \cite{shreve.soner.94}), both in accordance with the numerical observations of Constantinides \cite{constantinides.86}. With binding constraints, this no longer holds true: Both effects turn out to be proportional to the square-root $\epsilon^{1/2}$ of the spread, making precise the observation of \cite{dai.al.11} that transaction costs can have a first-order effect on welfare. The reason is that investors, if unconstrained, would like to hold larger risky positions then they are allowed to and therefore are more reluctant to tolerate downward swings of their risky position. This leads to a smaller no-trade region and hence, due to larger trading costs, to a bigger welfare impact.

More formally, this can also be rationalized by adapting the argument of Rogers~\cite{rogers.01}. He argued that the losses per unit time due to transaction costs are proportional to the local time of a reflected Brownian motion and hence of order $\epsilon/x$, where $x$ denotes the length of the no-trade region. This reasoning only depends on the width of the no-trade region and hence remains valid also in the presence of constraints. What changes is the displacement loss due to deviating from the frictionless optimal policy. In the absence of constraints, the optimal frictionless Merton proportion $\pi_{*}$ lies in the no-trade region $[\pi_-,\pi_+]$.\footnote{Note that this need not hold for leveraged positions if transaction costs are sufficiently large \cite[p.\ 675]{shreve.soner.94}, but is always true if the transaction costs are small enough (compare \cite[p.\ 182]{janecek.shreve.04}).} Consequently, the value function behaves like $(\pi-\pi_{*})^2$ near the maximum $\pi_{*}$ according to Taylor's theorem, leading to losses of order $x^2$ per unit time due to suboptimal portfolio composition. Minimizing the total loss then leads to $x\sim \epsilon^{1/3}$ and a welfare loss $\sim\epsilon^{2/3}$. In the presence of binding constraints, the frictionless minimizer no longer lies in the no-trade region; consequently, the value function behaves linearly leading to losses of order $x$. Again minimizing the total loss gives $x\sim\epsilon^{1/2}$ and a welfare loss of the same order.

\begin{figure}
\subfigure{\includegraphics[width=0.49\textwidth]{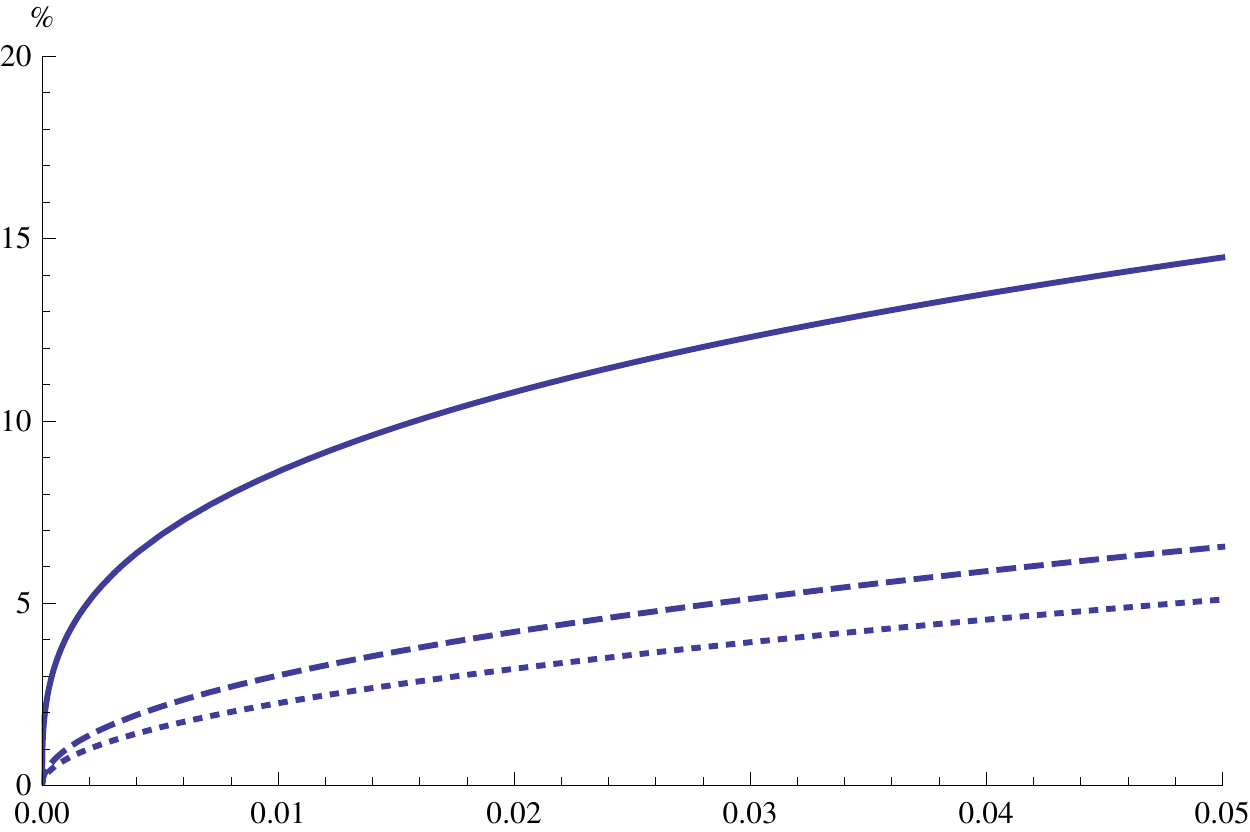}}\hfill
\subfigure{\includegraphics[width=0.49\textwidth]{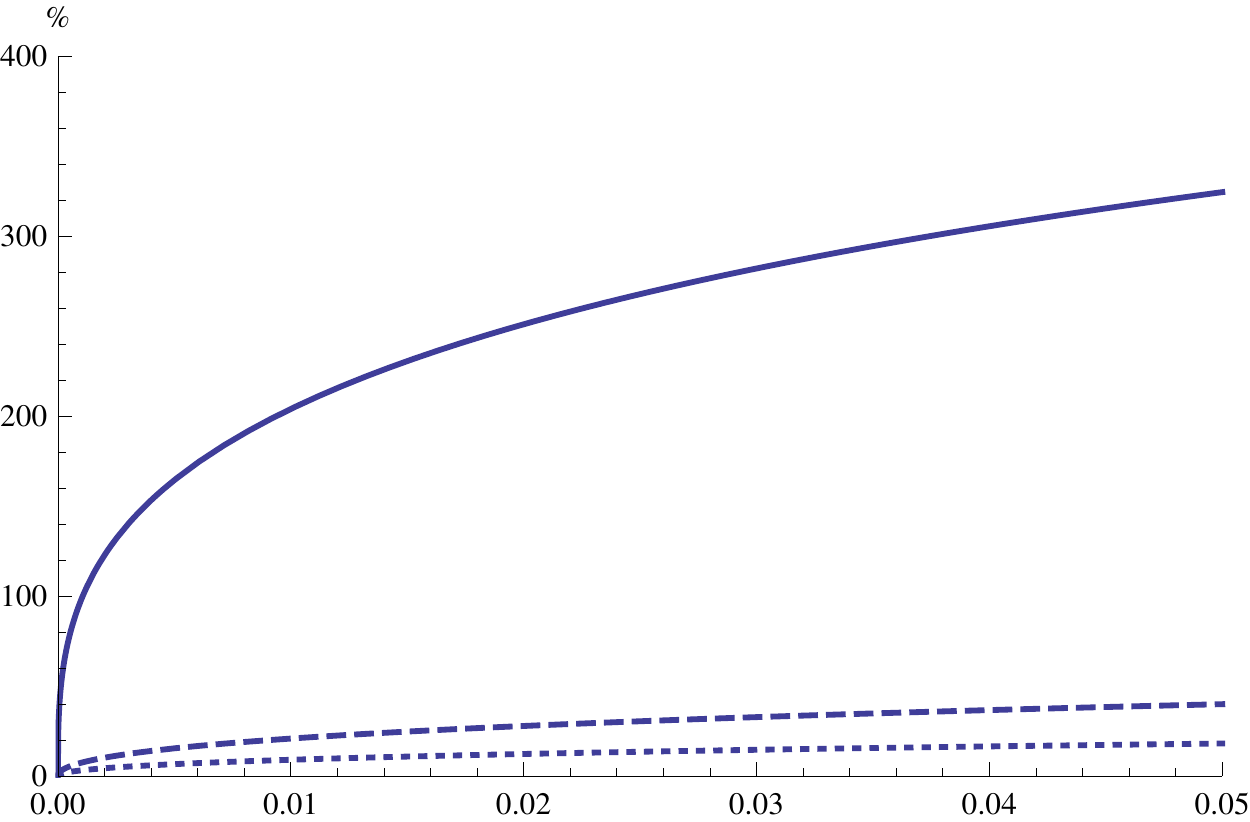}}
\caption{Left panel: Length of the no-trade region (vertical axis) plotted against the bid-ask spread $\epsilon$ (horizontal axis) for an unconstrained weight $\pi_*=62.5\%$ without leverage (solid), and with constraints $\pi_{\max}$ = 50\% (dashed) and 40\% (dotted). Right panel: Length of the no-trade region against the spread for an unconstrained weight of $\pi_*=390\%$ with leverage (solid), and with constraints $\pi_{\max}=225\%$ (dashed) and $175\%$ (dotted).  Model parameters are $\mu = 8\%$, $\sigma = 16\%$ and risk aversion is $\gamma=5$ resp.\ $\gamma=0.8$.  }
\label{fig:nt}
\end{figure}

For small transaction costs, it is possible to determine the comparative statics of the width of the no-trade region for varying constraint levels by analyzing the leading term in the asymptotic expansion. The above argument that investors constrained from above only tolerate smaller downward moves of their risky weight suggests that harder constraints should lead to a smaller no-trade region. This intuitive reasoning is indeed true almost generically, i.e., unless the unconstrained frictionless risky weight lies too close to unity.\footnote{In the degenerate case $\pi_*=1$, it is optimal to invest in a full risky position initially, and then hold the latter without further trades. Hence, the unconstrained no-trade region vanishes, but its width is increased by any non-trivial constraints.  For unconstrained frictionless weights close enough to one ($\pi_* \in [0.93,1.25]$) imposing constraints can similarly still increase the width of the no-trade region.} In all other cases the intuition that harder constraints lead to a smaller no-trade region indeed holds true at the leading order. Using the exact solution for $\lambda$ Figure~\ref{fig:nt} depicts the length of the no-trade region against the bid-ask spread and confirms the asymptotic result.

\subsection{Trading Volume}

In the frictionless case, the impact of binding constraints is simple: They reduce the investor's risky position. The present setting with transaction costs also allows to quantify their effect on the implied trading volume, which shows a more involved picture. Indeed, two competing effects are at work here: On the one hand, the constraints reduce the investor's risky position, thereby also reducing the amount of rebalancing necessary to keep the risky weight in the no-trade region. On the other hand, additional trading is required as the constraints (typically) lead to a smaller no-trade region. Our asymptotic results show that as the spread becomes small the latter effect prevails: The smaller no-trade region leads to share and wealth turnovers of order $\epsilon^{-1/2}$, in contrast to the rate $\epsilon^{-1/3}$ observed in the unconstrained case \cite{gerhold.al.11}. That is, for sufficiently small spreads, investors trade more in the presence of binding constraints. 

The comparative statics of share turnover for small costs can again be analyzed by means of the leading term in the asymptotic expansion. In the no-leverage case $\pi_* \in (0,1)$, harder constraints not only (typically) decrease the width of the no-trade region but also move the risky weight away from the degenerate buy-and-hold strategy obtained for $\pi_*=1$. Consequently, at the leading order $\epsilon^{-1/2}$, harder constraints always imply increased share turnover in this case. In contrast, the situation is ambiguous in the leverage case. On the one hand, harder constraints decrease the width of the no-trade region, thereby increasing turnover. On the other hand, however, they decrease turnover by pulling the risky weight closer to the buy-and-hold level. For sufficiently hard constraints ($\pi_{\max} \in [1,\frac{1+2\pi_{*}}{3})$) the latter effect prevails, decreasing turnover, and vice versa for $\pi_{\max}>\frac{1+2\pi_{*}}{3}$.

Since all of these results only hold true asymptotically as the spread $\epsilon$ becomes small, it is interesting to compare them with their exact counterparts, obtained by numerically solving for the gap $\lambda$. The results are reported in Figure \ref{fig:shtu}. In the absence of leverage for the frictionless weight $\pi_*$, there is perfect agreement with the asymptotic results: Turnover is increasing with harder constraints, and larger than in the unconstrained case. In the leverage case, however, the situation is less clear-cut. The comparative statics again match the asymptotic results: Both constraints lie in the domain where harder constraints decrease turnover asymptotically, matching the numerical results. Compared to the unconstrained case, however, we observe that very low levels of transaction costs may be needed for the constrained turnover to surpass its unconstrained counterpart, in particular, for tight constraints.

\begin{figure}
\subfigure{\includegraphics[width=0.49\textwidth]{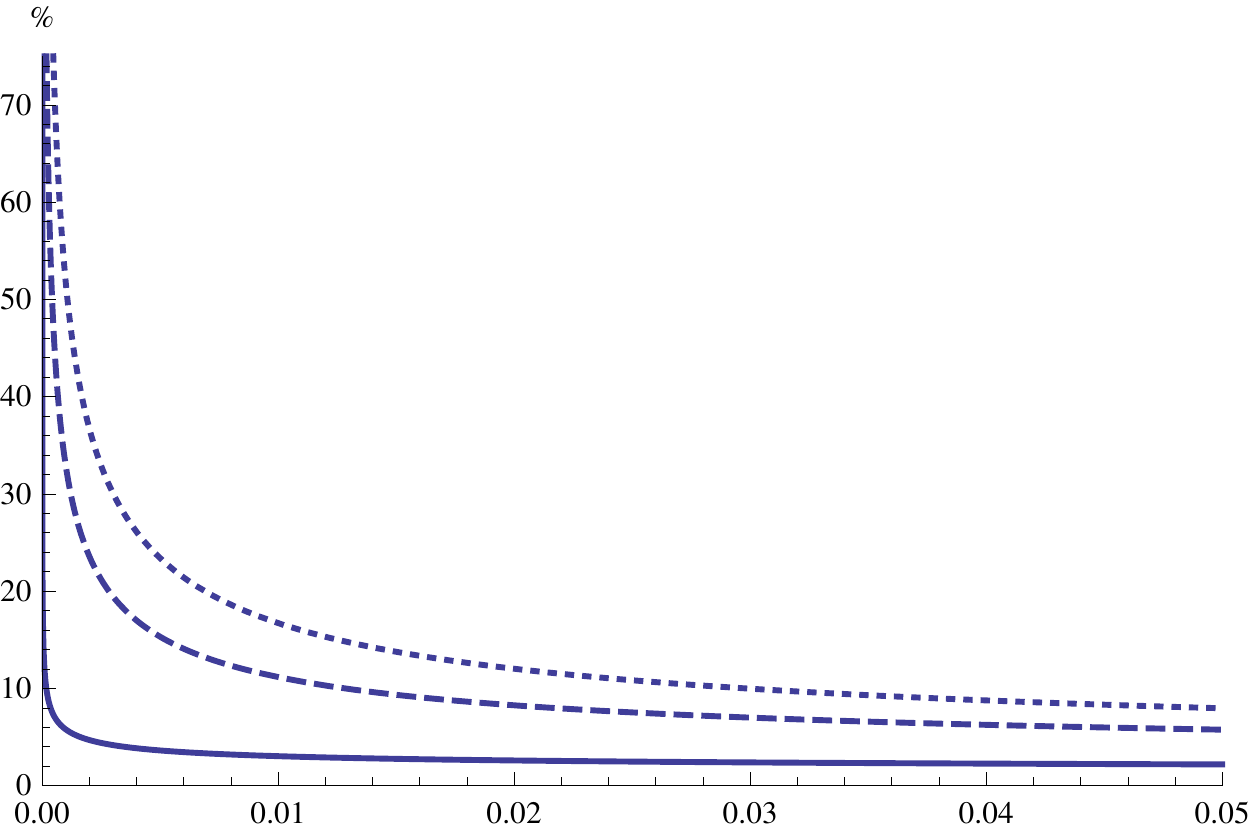}}\hfill
\subfigure{\includegraphics[width=0.49\textwidth]{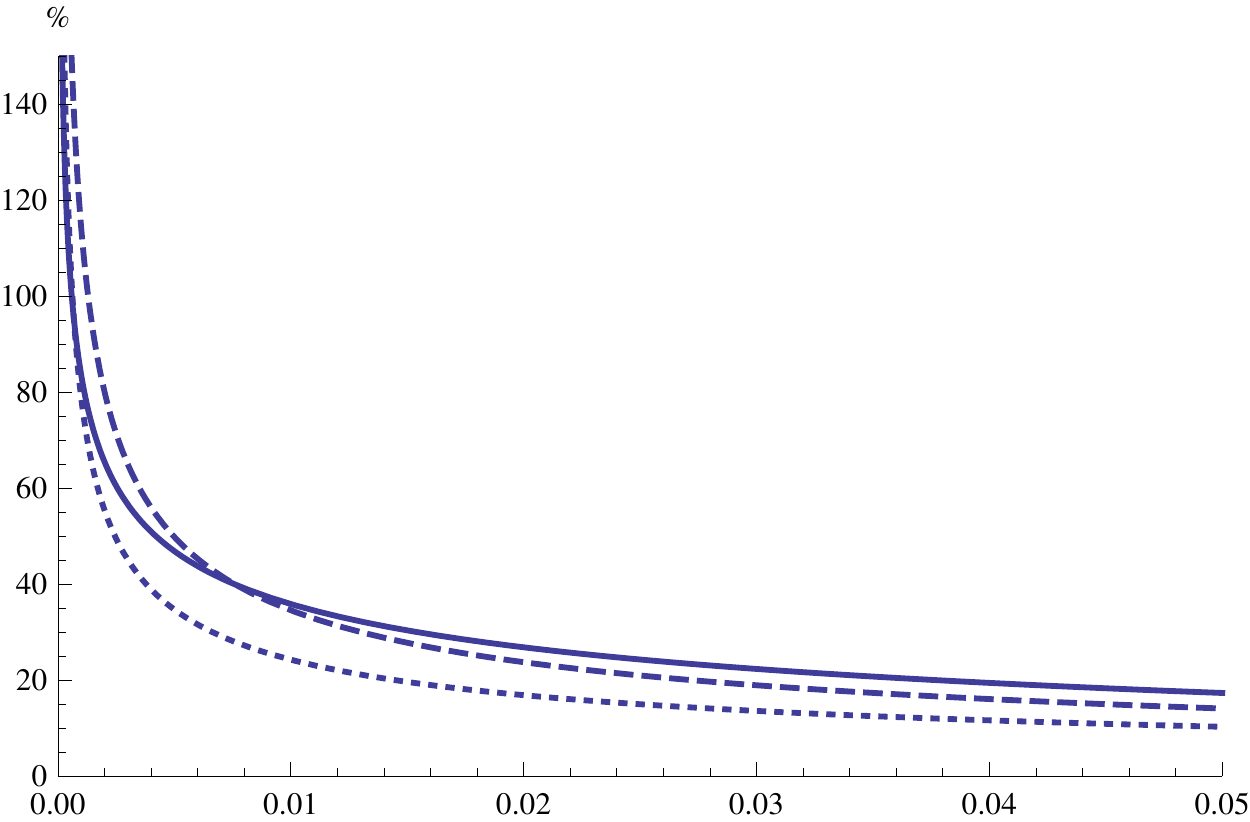}}
\caption{Left panel: Share turnover (vertical axis, annual fractions traded) plotted against the bid-ask spread $\epsilon$ (horizontal axis) for an unconstrained weight $\pi_*=62.5\%$ without leverage (solid), and with constraints $\pi_{\max}$ = 50\% (dashed) and 40\% (dotted). Right panel: Share turnover against the spread for an unconstrained weight of $\pi_*=390\%$ with leverage (solid), and with constraints $\pi_{\max}=225\%$ (dashed) and $175\%$ (dotted).  Model parameters are $\mu = 8\%$, $\sigma = 16\%$ and risk aversion is $\gamma=5$ resp.\ $\gamma=0.8$. }
\label{fig:shtu}
\end{figure}

The corresponding results for wealth turnover are more involved and are omitted for brevity.

\subsection{Turnover, Spreads, and Liquidity Premia}\label{relation}
In a model with transaction costs but without constraints, Gerhold et al.~\cite{gerhold.al.11} pointed out the following connection between the welfare impact of small transaction costs and the turnover implied by the optimal strategy:\footnote{Here, both quantities have to be measured consistently, either focusing on the risky asset (liquidity premium and share turnover) or on the whole market (equivalent safe rate and wealth turnover).} 
\begin{equation}\label{connection}
\left(r+ \frac{\mu^2}{2 \gamma \sigma^2}\right)-\esr \sim \frac{3}{4}\epsilon \wet  \quad  \text{and} \quad \lip \sim  \frac{3}{4}\epsilon \sht .
\end{equation} 
The interpretation is that the unobservable welfare effect of small transaction costs is approximately equal to a product of observables: trading volume, times the bid-ask spread, times a universal constant. This shows that the comparative statics of both quantities coincide, and allows to estimate liquidity premia from data on trading volume.

In the presence of binding constraints, the asymptotic rates of all involved quantities change. On the contrary, the link \eqref{connection} between them remains valid, up to changing the constants. To see this, first note that in the presence of the constraint but without transaction costs, the equivalent safe rate is given by $r+\frac{\mu^2}{2\gamma\sigma^2}(\frac{2\pi_{\max}}{\pi_*}-(\frac{\pi_{\max}}{\pi_*})^2)$. The additional reduction due to small transaction costs is therefore approximately equal to:
$$r+\frac{\mu^2}{2\gamma\sigma^2}\left(\frac{2\pi_{\max}}{\pi_*}-\left(\frac{\pi_{\max}}{\pi_*}\right)^2\right)-\esr \sim  \gamma\sigma^2 \left(\frac{1}{\gamma} (\pi_*-\pi_{\max})(1-\pi_{\max})^2\pi_{\max}^2\right)^{1/2}\epsilon^{1/2}.$$
In view of the asymptotic expansion for wealth turnover, this shows that the extra impact of transaction costs on the equivalent safe rate remains proportional to wealth turnover times the spread, also in the constrained case:
$$r+\frac{\mu^2}{2\gamma\sigma^2}\left(\frac{2\pi_{\max}}{\pi_*}-\left(\frac{\pi_{\max}}{\pi_*}\right)^2\right)-\esr \sim \epsilon \wet.$$
Similarly, notice that the total liquidity premium $\lip$ in Theorem \ref{main result} can be asymptotically decomposed into the liquidity premia $\lip^{\mathrm{C}}$ and $\lip^{\mathrm{T}}$ required to compensate for the constraints alone and the additional effect of the transaction costs:
\begin{displaymath}
\lip \sim  \mu\left(1-\sqrt{2\tfrac{\pi_{\max}}{\pi_*}-(\tfrac{\pi_{\max}}{\pi_*})^2}\right)+\gamma\sigma^2\left(\tfrac{1}{\gamma}\tfrac{\pi_{\max}(\pi_*-\pi_{\max})(1-\pi_{\max})^2}{2\pi_{*}-\pi_{\max}}\right)^{1/2}\epsilon^{1/2}:=\lip^\mathrm{C}+\lip^\mathrm{T}.
\end{displaymath}
With this notation, we obtain the following analogue of the second relation in \eqref{connection}:
$$ 
\lip^\mathrm{T} \sim \left(\frac{\pi_{\max}}{2\pi_*-\pi_{\max}}\right)^{1/2} \epsilon \sht.
$$
Hence, this result is also robust to the additional portfolio constraints, up to one important caveat. Unlike for wealth turnover and the equivalent safe rate above, the constant linking share turnover and the liquidity premium accrued due to transaction costs depends on the constraints, thereby leading to different comparative statics. Indeed, whereas the leading term of share turnover is always increasing with harder constraints in the absence of leverage, $\pi_{*} \in (0,1)$, the effect on the liquidity premium can be ambiguous due the presence of the extra factor $[\pi_{\max}/(2\pi_*-\pi_{\max})]^{1/2}$, which is decreasing with harder constraints. This is in line with the numerical observation of \cite{dai.al.11} that ``the liquidity premium can be higher even though position limits are less binding''. Whereas this may or may not be the case in the absence of leverage, it is in fact the generic situation for a leveraged position $\pi_*>1$ and constraints $\pi_{\max} \in [1,\frac{1+2\pi_{*}}{3})$.

\begin{figure}
\subfigure{\includegraphics[width=0.49\textwidth]{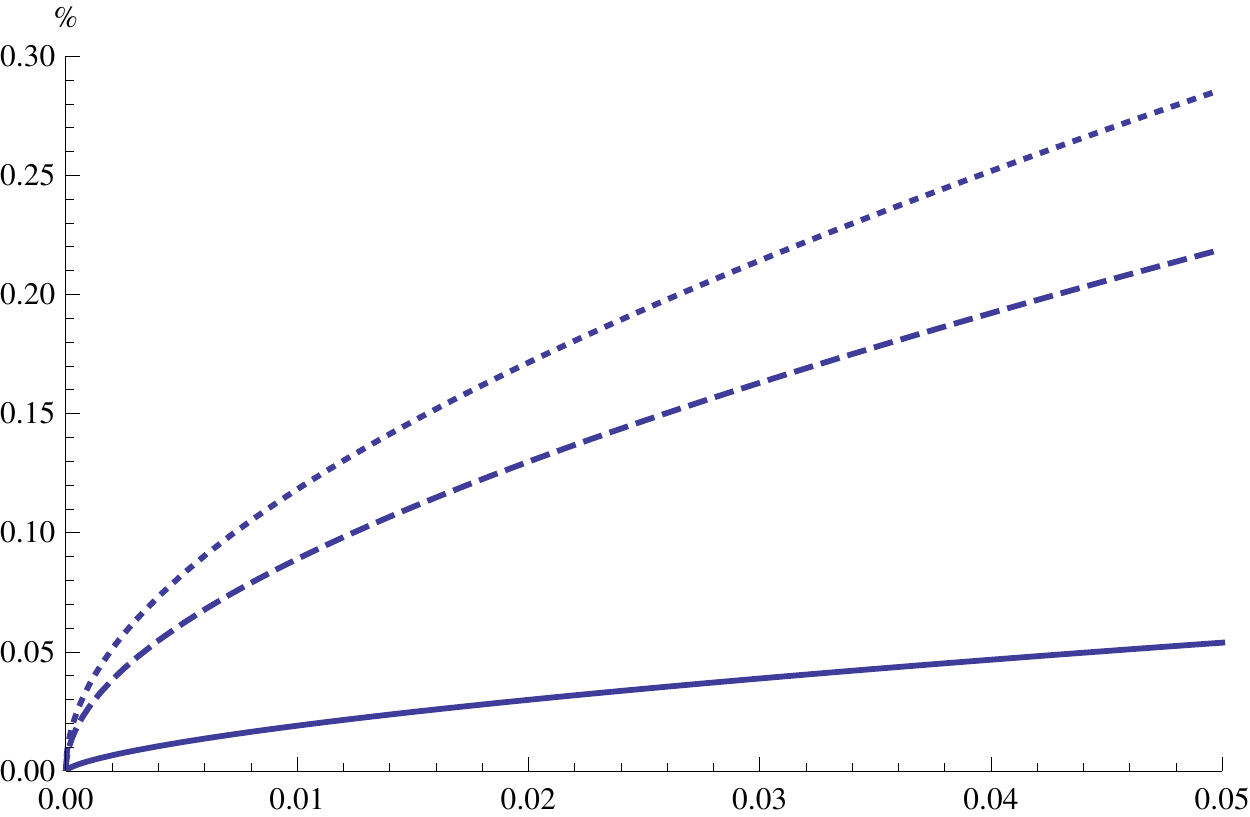}}\hfill
\subfigure{\includegraphics[width=0.49\textwidth]{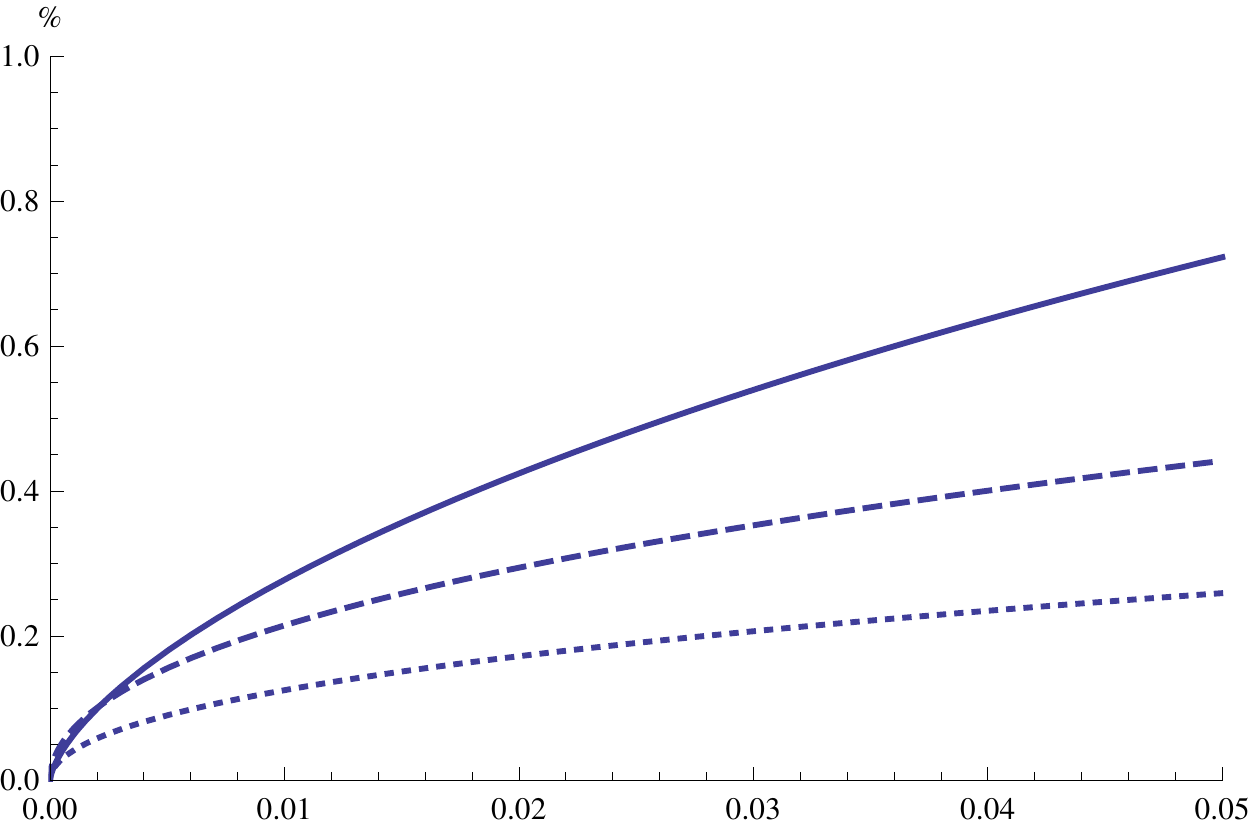}}
\caption{Left panel: Liquidity premium for transaction costs $\lip^\mathrm{T}$ (vertical axis) plotted against the bid-ask spread $\epsilon$ (horizontal axis) for an unconstrained weight $\pi_*=62.5\%$ without leverage (solid), and with constraints $\pi_{\max}$ = 50\% (dashed) and 40\% (dotted). Right panel: $\lip^\mathrm{T}$ against the spread for an unconstrained weight of $\pi_*=390\%$ with leverage (solid), and with constraints $\pi_{\max}=225\%$ (dashed) and $175\%$ (dotted).  Model parameters are $\mu = 8\%$, $\sigma = 16\%$ and risk aversion is $\gamma=5$ resp.\ $\gamma=0.8$. }
\label{fig:lipr_tac}
\end{figure}

These results are illustrated in Figure \ref{fig:lipr_tac}, where the liquidity premia $\lip^\mathrm{T}$ due to transaction costs are plotted against the spread in the unconstrained case and for two binding constraints. In the no-leverage regime depicted in the left panel, the constrained liquidity premia dominate their unconstrained counterparts for all levels of transaction costs, in line with the larger asymptotic rate. Moreover, the liquidity premium increases with tighter constraints, in accordance the asymptotic comparative statics. These also predict the correct effect in the leverage case reported in the right panel. Here,  $\pi_{\max} \in [1,\frac{1+2\pi_{*}}{3})$  such that the leading-order terms of the liquidity premia are decreasing with tighter constraints, which matches the numerical results. Analogously as for share turnover, however, the constrained liquidity premium only dominates the unconstrained one for sufficiently low transaction costs and not for arbitrary levels like in the unleveraged case.

\begin{figure}
\subfigure{\includegraphics[width=0.49\textwidth]{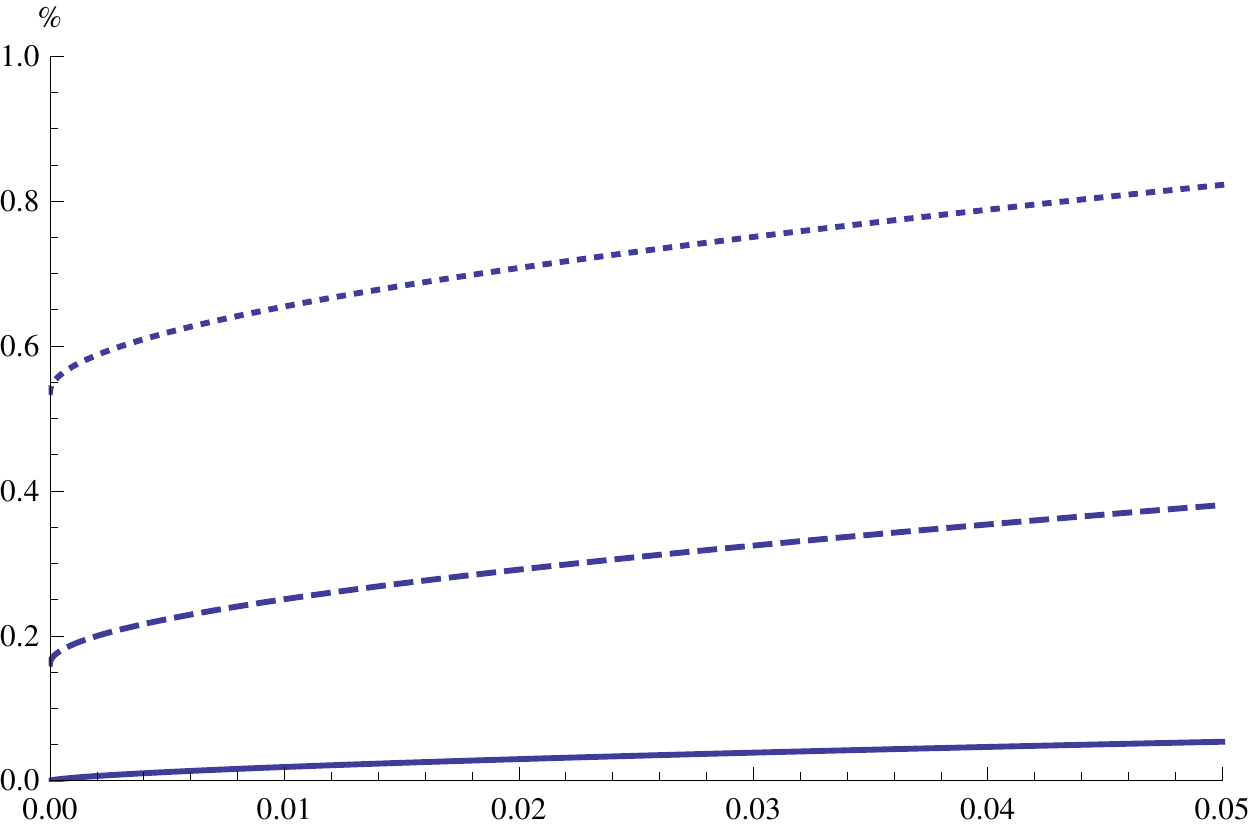}}\hfill
\subfigure{\includegraphics[width=0.49\textwidth]{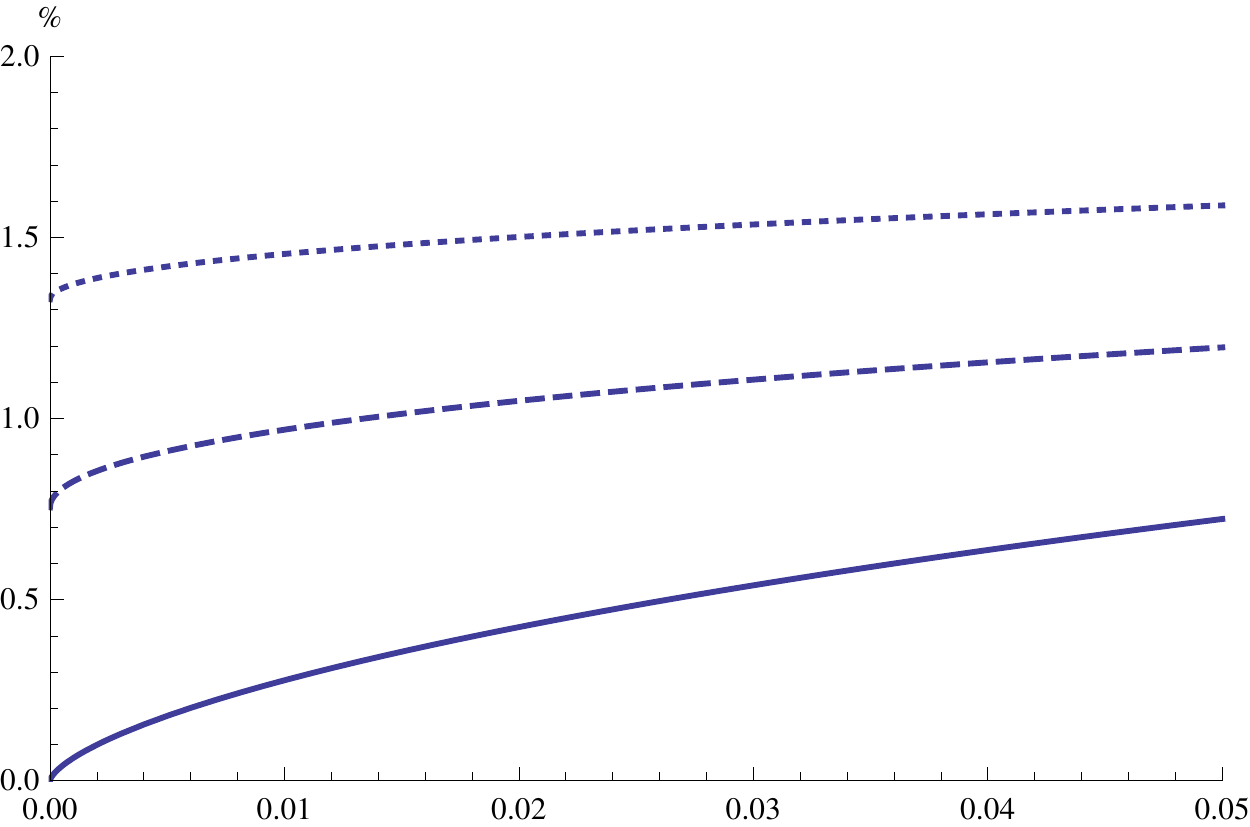}}
\caption{Left panel: Total liquidity premium $\lip$ (vertical axis) plotted against the bid-ask spread $\epsilon$ (horizontal axis) for an unconstrained weight $\pi_*=62.5\%$ without leverage (solid), and with constraints $\pi_{\max}$ = 50\% (dashed) and 40\% (dotted). Right panel: $\lip$ against the spread for an unconstrained weight of $\pi_*=390\%$ with leverage (solid), and with constraints $\pi_{\max}=225\%$ (dashed) and $175\%$ (dotted).  Model parameters are $\mu = 8\%$, $\sigma = 16\%$ and risk aversion is $\gamma=5$ resp.\ $\gamma=0.8$. }
\label{fig:lipr_total}
\end{figure}

Since the total liquidity premium $\lip$ measures the joint impact of constraints and transaction costs, it is interesting to discriminate between the relative contributions of these two market frictions. Figure \ref{fig:lipr_total} depicts the total liquidity premia for the same examples as before. Evidently, the impact of the transaction costs is larger in the leverage case, as in the absence of constraints~\cite{gerhold.al.11}, but the effect is diminished severely as the constraints bind harder. In contrast, the transaction costs have a much bigger influence than in the unconstrained case in the absence of leverage.

\subsection{Applications}
\subsubsection*{Selection of Prime Brokers}

Consider an investor choosing which \emph{prime broker} to use to buy a leveraged risky position \emph{on margin}. Each broker is willing let the investor borrow from him at a lending rate $r>0$ and up to a leverage constraint $\pi_{\max}>1$,\footnote{Put differently, the broker's \emph{margin requirement} is $1/\pi_{\max}$.} thereby allowing the investor to trade in a market with safe rate $r$, excess return $\bar\mu-r$ of the risky asset ($\bar\mu$ denotes the total return), and leverage constraint $\pi_{\max}$.

Without transaction costs, the same equivalent safe rate $\mathrm{esr}$ can be achieved if
\begin{equation}\label{eq:endr}
r=\frac{2\pi_{\max}\bar\mu-2\mathrm{esr}-\pi_{\max}^2 \gamma\sigma^2}{2(\pi_{\max}-1)}.
\end{equation}
That is, all pairs $(r,\pi_{\max})$ satisfying this relation lie on the same iso-utility curve for the investor, since the effect of having to pay a higher lending rate is precisely offset by the opportunity to borrow a larger amount for investing. This is illustrated in Figure \ref{fig:esr1}.

With transaction costs, the investor is no longer indifferent between these different brokers, and our model allows to asses which combinations are more attractive if the respective rebalancing costs are taken into account. Indeed, for pairs $(r,\pi_{\max})$ satisfying \eqref{eq:endr}, the equivalent safe rate with transaction costs is given by
$$\mathrm{esr}-\left(\frac{\sigma^2}{2}\pi_{\max}^2(\pi_{\max}-1)(2\mathrm{esr}-2\bar\mu-(\pi_{\max}-2)\pi_{\max}\gamma\sigma^2)\right)^{1/2}\epsilon^{1/2}+O(\varepsilon).$$
If $\pi_{\max}=1$, the investor follows a buy-and-hold strategy and the transaction costs have no impact. As the constraint becomes softer, the effect of transaction costs increases, reaches its maximum at a critical level $\pi_{\max}^\mathrm{c}$ and then decreases again, vanishing at the leading order $\epsilon^{1/2}$ as the constraint $\pi_{\max}$ tends towards the frictionless Merton proportion $\pi_*=(\bar\mu-r)/\gamma\sigma^2$.\footnote{For this value, the welfare impact of transaction costs is only of order $\epsilon^{2/3}$ as in the unconstrained case \cite{gerhold.al.11}.} Consequently, for sufficiently hard constraints, the investor prefers tighter constraints and lower lending rates to softer constraints and higher lending rates, but the picture is reversed for sufficiently soft constraints, as illustrated in Figure \ref{fig:esr1}.

\begin{figure}
\subfigure{\includegraphics[width=0.49\textwidth]{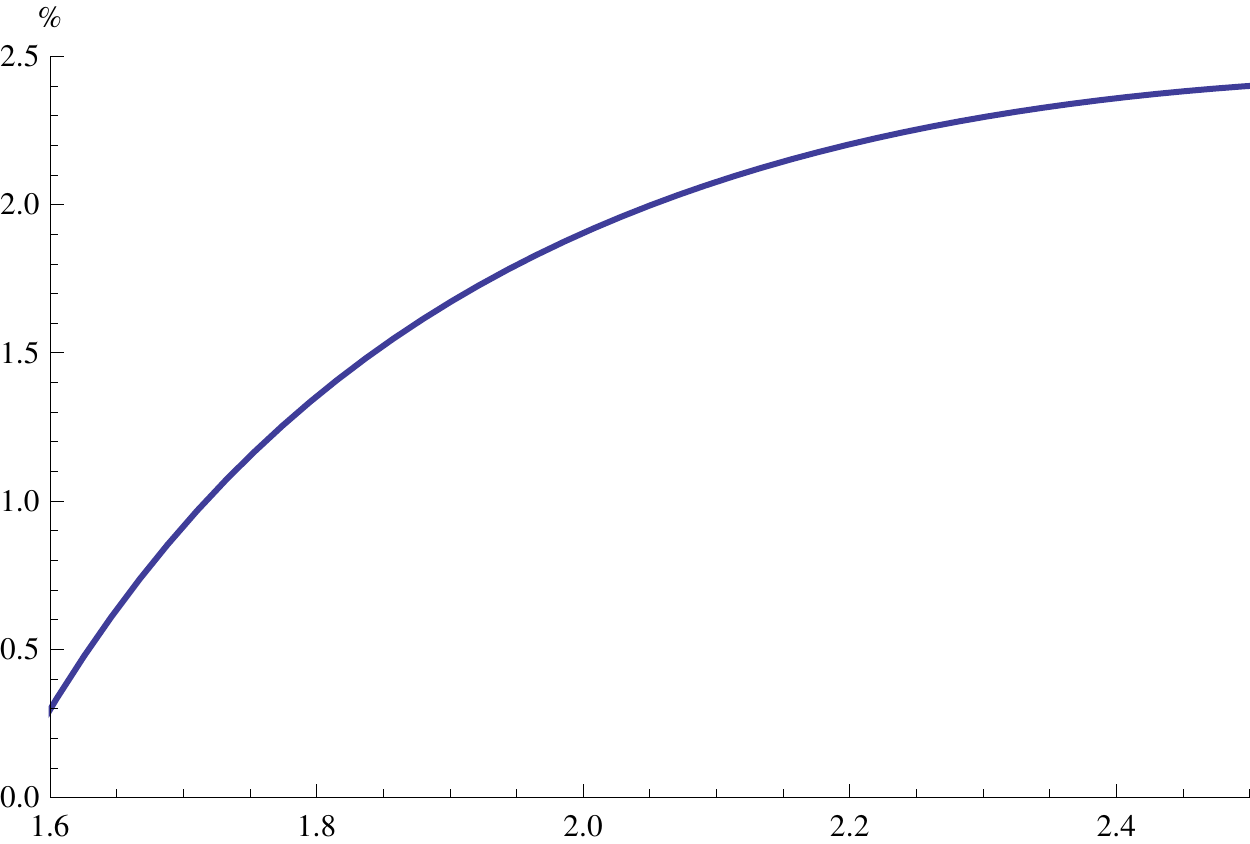}}
\subfigure{\includegraphics[width=0.49\textwidth]{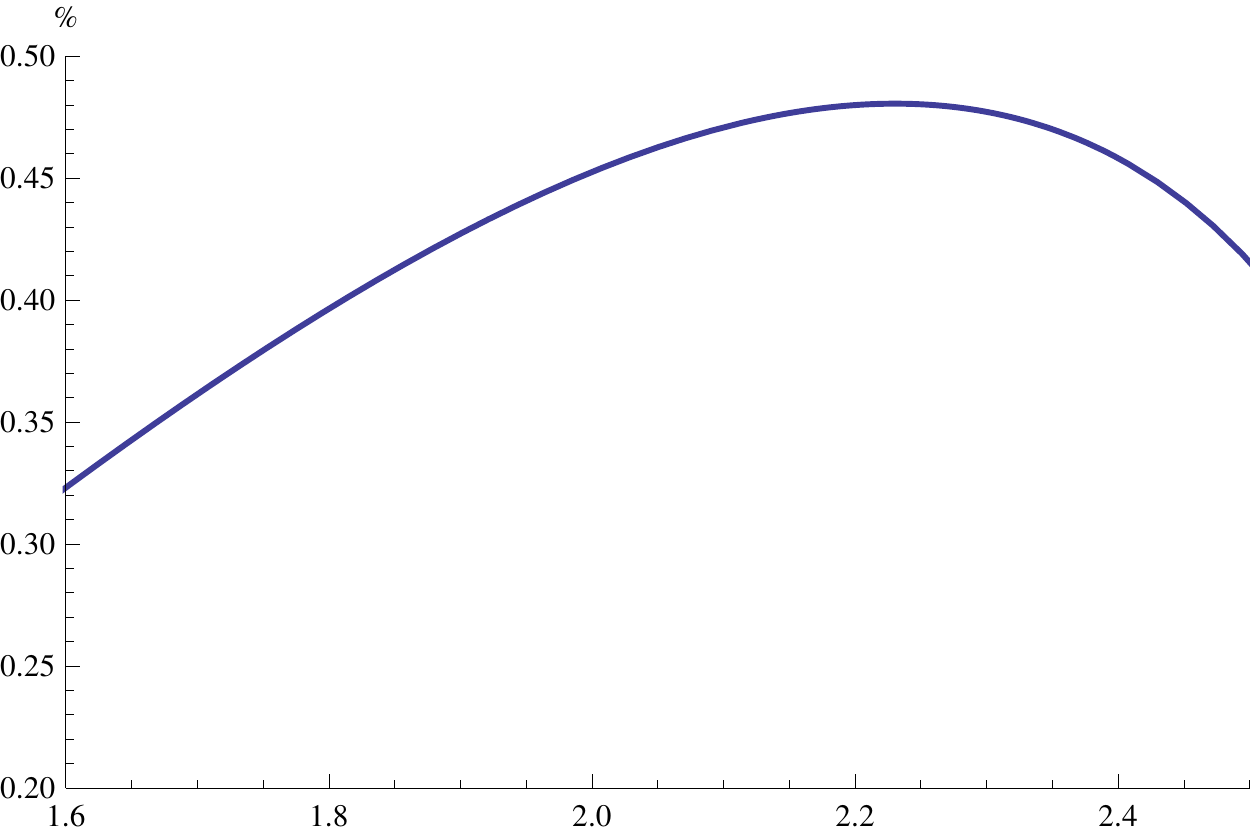}}\hfill
\caption{Left panel: $10\%$ equivalent safe rate indifference curve without transaction cost plotted against the interest rate $r$ (vertical axis) and the constrained weight $\pi_{\max}$ (horizontal axis, the corresponding unconstrained weight is $271\%$). Right panel: Leading-order loss in equivalent safe rate due to transaction costs $\epsilon=1\%$ plotted against the constrained weight $\pi_{\max}$ with an interest rate $r$ such that the equivalent safe rate without transaction cost is always 10\% . Model parameters are $\bar\mu = 8\%$, $\sigma = 16\%$, and risk aversion is $\gamma=0.8$.}
\label{fig:esr1}
\end{figure}

\begin{figure}
\subfigure{\includegraphics[width=0.49\textwidth]{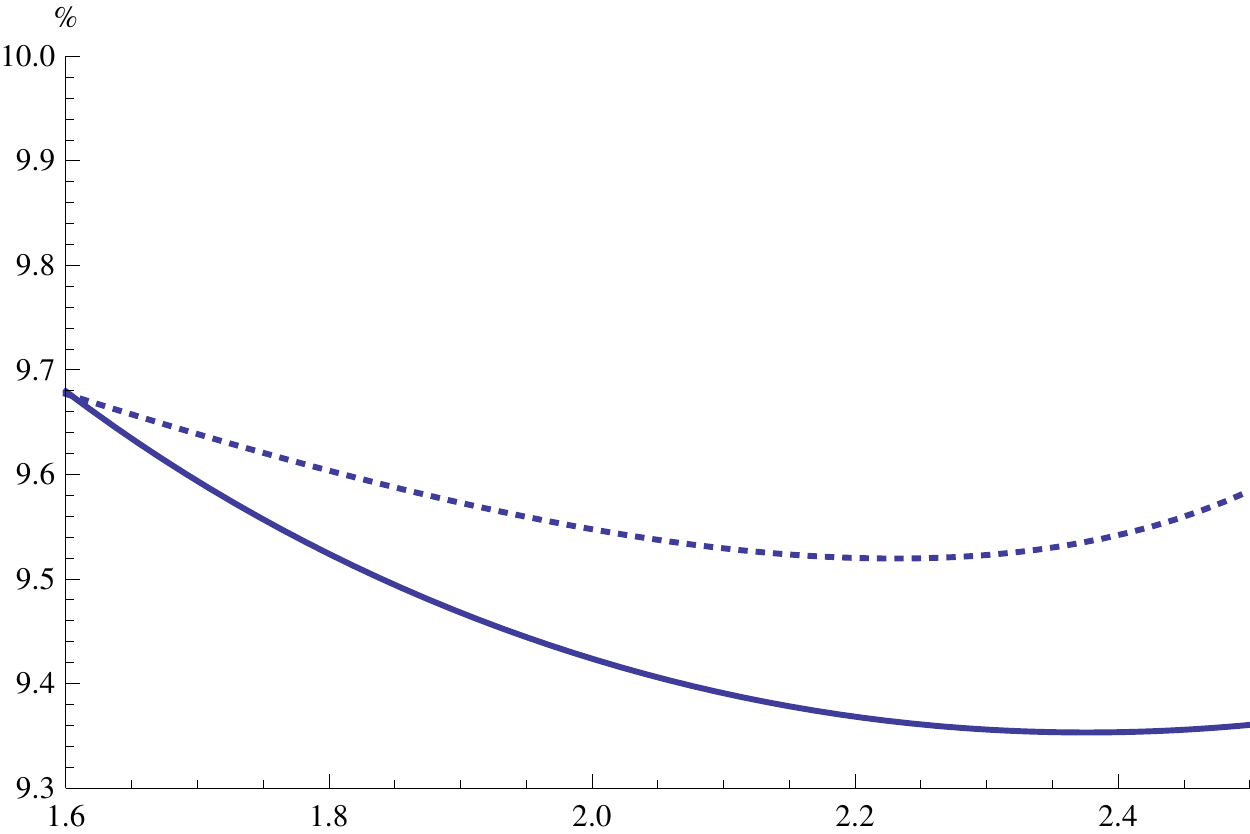}}
\subfigure{\includegraphics[width=0.49\textwidth]{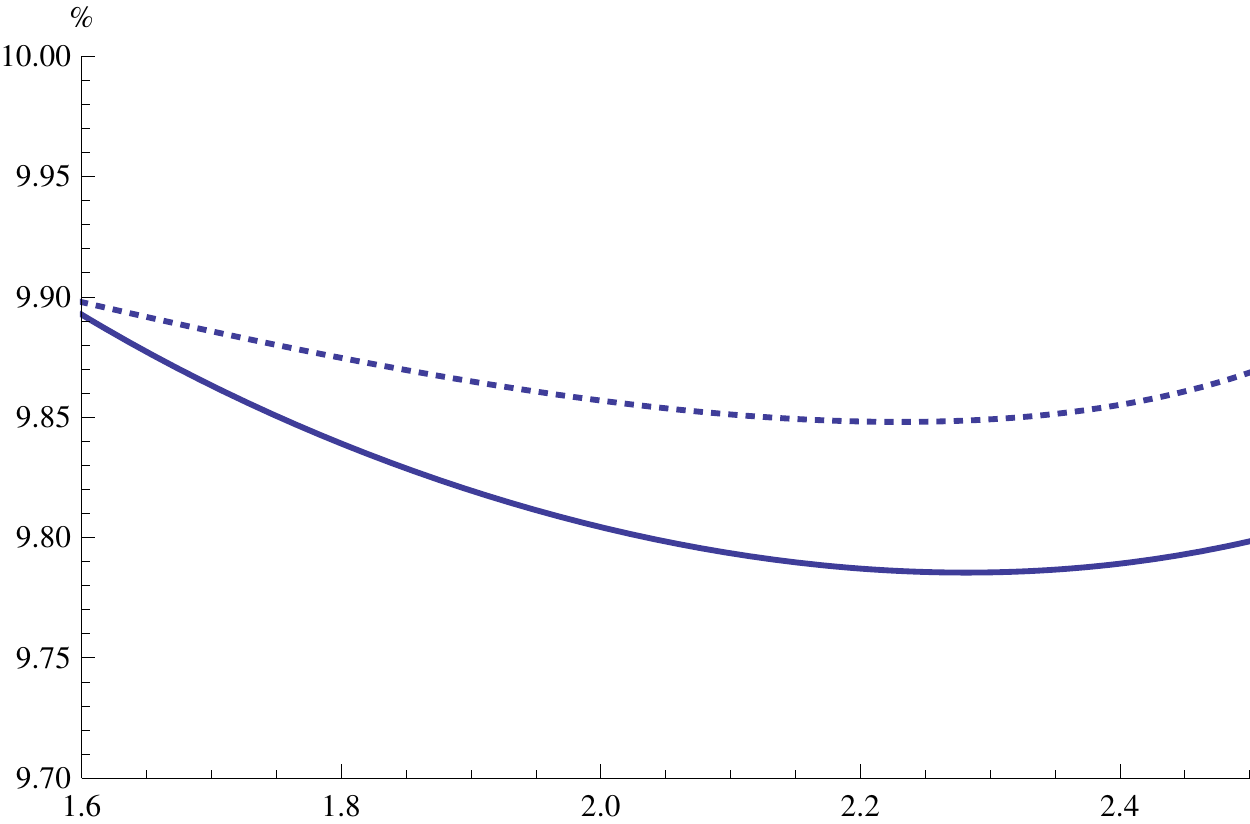}}\hfill
\caption{ Exact (solid) and approximate equivalent safe rate (dotted) with transaction costs (vertical axis, left panel: $1\%$, right panel: $0.1\%$) plotted against the constrained weight $\pi_{\max}$ with an interest rate $r$ such that the equivalent safe rate without transaction cost is always 10\% . Model parameters are $\bar\mu = 8\%$, $\sigma = 16\%$, $\pi_{\max} = 271\%$ and risk aversion is $\gamma=0.8$.}
\label{fig:esr2}
\end{figure}

The above leading-order analysis neglects that the effect of transaction costs does not vanish exactly as the constraint approaches the Merton proportion; only the leading term of order $\epsilon^{1/2}$ tends to zero. More precisely, it is replaced by a nontrivial term of order $\epsilon^{2/3}$ in the limit. Hence, portfolios with soft constraints are in fact less attractive than suggested by the above analysis. This is illustrated in Figure \ref{fig:esr2}, which compares the exact and leading-order equivalent safe rates that can be obtained with pairs $(r,\pi_{\max})$ satisfying \eqref{eq:endr}. Whereas the approximation recaptures the qualitative properties of the exact quantity rather well, it severely overestimates the attractiveness of brokers with a soft leverage constraint (i.e., a low margin requirement) and a high lending rate. With the exact quantities, it turns out that the investor typically\footnote{That is, unless the constraints are barely binding.} prefers tighter leverage constraints and lower lending rates, as transaction costs make highly leveraged positions relatively less attractive than they appear to be in frictionless markets.

\subsubsection*{Illiquid Loans and Deposit Rates}

The application in the previous section can also be reinterpreted as follows. Consider a bank, who can borrow from its depositors at a safe rate $r$ to provide illiquid (long-term) loans, whose book values are assumed to follow geometric Brownian motion with constant drift $\bar \mu$ and volatility $\sigma$. To limit excessive risk-taking, regulating authorities restrict the amount of leverage financial institutions are allowed to use by setting minimal capital requirements, which correspond to the portfolio constraints in our model. 

This setting allows to study how the deposit rates offered by the bank react to harder regulatory constraints. The idea is that the bank will try to achieve the same performance (measured in terms of the equivalent safe rate) with the new constraints. Since tightening the constraints reduces the equivalent safe rate, this means that the bank will decrease its deposit rate. In the absence of transaction costs, i.e., for loans that can be liquidated at their book values, the new deposit rate can be obtained using Formula \eqref{eq:endr}. However, most long-term loans are not very liquid, incurring substantial transaction costs when liquidated prior to maturity. The corresponding change in the deposit rate can in turn be determined by numerically solving for the safe rate in Theorem \ref{main result} that makes the equivalent safe rate with the new constraint coincide with its counterpart for the old constraint and the old deposit rate. 

The results for a concrete example are provided in Table \ref{tab:deposit rate}. In the literature, banks are often modeled as risk-neutral, simply maximizing the present value of future cash flows. Therefore, we also use a low risk aversion ($\gamma=0.1$) here. As the transaction costs incurred when prematurely liquidating long-term loans are substantial (in particular, compared to the bid-ask spreads observed for equities), we report results for a relative bid-ask spread of up to $10\%$. It turns out that harder regulatory constraints decrease the deposit rates the most if the long-term loans the bank provides are assumed to be perfectly liquid. If illiquidity is taken into account, then highly leveraged positions are less attractive, and a substantially smaller reduction of the deposit rate is required to compensate the bank for the tighter constraints.

\begin{table}
\begin{center}
\begin{tabular}{llll}
\hline
Constrained weights $\pi_{\max}$ & Transaction Costs $\epsilon$ & Deposit Rates $r$ &corresponding ESR\\
\hline
\multirow{4}{*}{$220 \%$} & $0\%$ & $5.82\%$ & $10\%$ \\
 & $0.1\%$ & $5.82\%$ & $9.84\%$\\
 & $1\%$ & $5.82\%$ & $9.54\%$\\
 & $10\%$ & $5.82\%$ &$8.86\%$ \\
\hline
\multirow{4}{*}{$180 \%$} & $0\%$ &$4.98\%$ & $10\%$ \\
 & $0.1\%$ & $5.05\%$ &$9.84\%$\\
 & $1\%$ &$5.16\%$ & $9.54\%$ \\
 & $10\%$ & $5.43\%$ & $8.86\%$\\
\hline
 \multirow{4}{*}{$150 \%$} & $0\%$ &$3.42\%$ & $10\%$ \\
 & $0.1\%$ & $3.60\%$ &$9.84\%$\\
 & $1\%$ &$3.93\%$ & $9.54\%$ \\
 & $10\%$ & $4.68\%$ & $8.86\%$\\
\hline
\end{tabular}
\end{center}
\caption{Changes in deposit rates $r$ due to a decrease of the upper bound of the risky weight from $220\%$ to $180\%$ resp. to $150\%$, in order to retain the initial level of the equivalent safe rate. Model parameters are $\bar \mu = 0.08$, $\sigma = 16\%$ and risk aversion is $\gamma = 0.1$.}
\label{tab:deposit rate}
\end{table}

\section{Heuristics}\label{heuristic}
In this section, we use methods from stochastic control to heuristically derive a candidate solution.


Fix an upper bound $0<\pi_{\max}=\kappa \pi_* \not=1$ on the investor's risky weight and consider the problem of maximizing the expected power utility $U(x) = x^{1-\gamma}/(1-\gamma)$ from terminal wealth at time $T$. Denote by $V(t,X_t^0,X_t)$ its value function, which is assumed to depend on time as well as the positions $X^0 = \varphi^0 S^0$ in the safe and $X= \varphi S$  in the risky asset, evaluated in terms of the ask price. Then, by It\^o's formula and the self-financing condition \eqref{eq:sf}:
\begin{align*}
dV(t,X^0_t,X_t)=&\left(V_t+rX^0_t V_x+(\mu+r)X_t V_y+\tfrac{\sigma^2}{2} X_t^2 V_{yy}\right)dt+\sigma X_t V_y dW_t\\
&\quad +S_t(V_y-V_x)d\varphi^\uparrow_t + S_t((1-\epsilon)V_x-V_y)d\varphi^\downarrow_t,
\end{align*}
where the arguments are omitted for brevity. By the martingale optimality principle, this process has to be a supermartingale for any strategy, and a martingale for the optimizer, leading to the HJB equation:
\begin{displaymath}
V_t + r X_t^0V_x + (\mu+r)X_t V_y+ \frac{\sigma^2}{2}X_t^2V_{y y}=0, \quad \text{if} \quad 1 < \frac{V_x}{V_y} < \frac{1}{1-\epsilon}.
\end{displaymath}
The homotheticity of the value function and the observation that -- in the long-run -- the value function should grow exponentially with the horizon at a constant rate $r+\beta$ suggest the following representation:\footnote{This representation is valid if the position in the safe asset is non-negative at all times. In the leverage case $\pi_{\max} > 1$ one has to factor out $-X_t^0$ instead of $X_t^0$, leading to analogous calculations.}
\begin{displaymath}
V(t,X_t^0,X_t) = (X_t^0)^{1-\gamma}v (X_t/X_t^0)e^{-(1-\gamma)(r+\beta)t}.
\end{displaymath}
Setting $z:= y/x$ the HJB equation becomes
\begin{displaymath}
\frac{\sigma^2}{2}z^2v''(z) + \mu z v'(z) - (1-\gamma)\beta v(z) = 0, \quad \text{if} \quad 1+z < \frac{(1-\gamma)v(z)}{v'(z)} < \frac{1}{1-\epsilon}+z.
\end{displaymath}
The set $\{z: 1+z < \frac{(1-\gamma)v(z)}{v'(z)} < \frac{1}{1-\epsilon}+z\}$ corresponds to those values of the stock-cash ratio $X/X^0$ for which the optimal strategy does not move, i.e., the no-trade region. To simplify further, assume that it is given by an interval $l< z < u$, where the lower boundary $l$ is an unknown parameter and the upper boundary $u$ coincides with the constraint, $u=\pi_{\max}/(1-\pi_{\max})$.\footnote{Without constraints and transaction costs, the investor would hold an even larger risky weight. Therefore it is natural to assume that the upper selling boundary coincides with the highest value compatible with the constraints.} Then, one obtains the following system with one free boundary:
\begin{eqnarray}
z^2v''(z)\sigma^2/2 + \mu z v'(z) - (1-\gamma)\beta v(z) &=& 0, \quad \text{if} \quad l < z < u,\label{ode}\\
(1+l)v'(l) -(1-\gamma)v(l) &=& 0,\label{freeboundary1}\\
\left(1/(1-\epsilon)+u\right)v'(u) -(1-\gamma)v(u) &=& 0.\label{freeboundary2}
\end{eqnarray}
These conditions do not suffice to identify the solution, since the ODE~\eqref{ode} is of order two and the conditions~\eqref{freeboundary1} and~\eqref{freeboundary2} can be matched for any choice of the buying boundary $l$. The optimal buying boundary $l$ is the one that additionally satisfies a smooth pasting condition \cite{dumas.91}, obtained by formally differentiating~\eqref{freeboundary1}:\footnote{Here, ``smooth'' means $C^2$ across the boundary. Note that this need not hold at the upper selling boundary, since the latter is fixed by the constraints.}
\begin{equation}\label{boundary l'}
(1+l)v''(l)+\gamma v'(l) = 0.
\end{equation}
Substituting~\eqref{boundary l'} and~\eqref{freeboundary1} into~\eqref{ode} yields
\begin{equation}\label{piequation}
-\frac{\sigma^2 \gamma}{2}\left(\frac{l}{1+l}\right)^2+\mu \frac{l}{1+l}-\beta = 0.
\end{equation}
The smaller solution of this quadratic equation determines the lower buying boundary:
\begin{displaymath}
\pi_- =\frac{l}{1+l}= \frac{\mu}{\gamma \sigma^2} - \frac{\sqrt{\mu^2-2\beta \gamma \sigma^2}}{\gamma \sigma^2}.
\end{displaymath}
Set $\pi_{-}=(1-\lambda)\pi_{\max}=(1-\lambda)\kappa \pi_*$ for some $\lambda>0$. Then, the growth rate $\beta$ can be written as
\begin{displaymath}
\beta = \frac{\mu^2}{2 \gamma \sigma^2}(1-(1-\kappa(1-\lambda))^2).
\end{displaymath}
$\lambda$ is called the \emph{gap}, because it describes the deviation of the frictional buying boundary and growth rate from their frictionless counterparts. Since the buying boundary is determined by $\lambda$, the above free-boundary value problem becomes a fixed-boundary value problem with free parameter $\lambda$ in this notation. The substitution
\begin{displaymath}
v(z)=e^{(1-\gamma)\int_0^{\log{(z/l(\lambda))}}w(y)d y}, \quad \text{i.e.,} \quad w(y)=\frac{l(\lambda)e^y v'(l(\lambda)e^y)}{(1-\gamma)v(l(\lambda)e^y)},
\end{displaymath}
in turn reduces it to a Riccati ODE: 
\begin{equation}\label{Riccati}
0 = w'(x) + (1-\gamma)w(x)^2 +(\tfrac{2\mu}{\sigma^2}-1)w(x) -\tfrac{\mu^2}{\gamma\sigma^4}(1-(1-\kappa(1-\lambda))^2),\qquad x\in \left[0,\log{(u/l(\lambda))}\right],
\end{equation}
with boundary conditions
\begin{align}
w(0) &= \frac{l(\lambda)}{1+l(\lambda)} = (1-\lambda)\pi_{\max},\label{w0x}\\
w\left(\log{\left(\frac{u}{l(\lambda)}\right)}\right) &= \frac{u(1-\epsilon)}{1+u(1-\epsilon)}=\frac{\pi_{\max}(1-\epsilon)}{(1-\pi_{\max})+\pi_{\max}(1-\epsilon)},\label{wend}
\end{align}
where
\begin{equation}\label{logul}
\frac{u}{l(\lambda)}=\frac{\pi_{\max}/(1-\pi_{\max})}{\pi_-(\lambda)/(1-\pi_{-}(\lambda))}=\frac{\pi_{\max}/(1-\pi_{\max})}{(1-\lambda)\pi_{\max}/(1-(1-\lambda)\pi_{\max})}.
\end{equation}
Since the Riccati ODE is of order one, the initial condition~\eqref{w0x} uniquely determines a solution $w(\lambda,\cdot)$ for any choice of $\lambda$. The correct one is then identified by the terminal condition~\eqref{wend}. Even though the Riccati ODE~\eqref{Riccati}-\eqref{w0x} can be solved explicitly (cf.\ Lemma~\ref{solw}), it is not possible to solve for $\lambda$ in closed form. However, the implicit function theorem readily yields a fractional power series expansion in $\varepsilon$ (cf.\ Lemma \ref{uniquelambda}), which in turn immediately provides the asymptotics for the buying boundary $\pi_{-}=(1-\lambda)\pi_{\max}$ and the growth rate $\beta=\frac{\mu^2}{2 \gamma \sigma^2}(1-(1-\kappa(1-\lambda))^2)$.

\section{Proofs}\label{proof}

\subsection{Explicit Formulae and their Properties}
The first step towards a rigorous verification theorem is to determine an explicit expression for the solution of the Riccati ODE~\eqref{Riccati} with initial condition~\eqref{w0x}, given a sufficiently small $\lambda>0$.
\begin{mylemma}\label{solw}
Let $0 < \pi_{\max}=\kappa\pi_* \not= 1$ and define $w(\lambda,\cdot)$ by
\begin{displaymath}
w(\lambda,x) :=
\begin{cases}
\frac{a(\lambda)\tanh{\left(\tanh^{-1}{(b(\lambda)/a(\lambda))}-a(\lambda)x\right)}+\mu/\sigma^2-1/2}{\gamma-1},\\
\qquad \qquad \qquad \qquad \text{if  } \gamma \in (0,1) \text{  and  }0 < \pi_{\max} <1, \text{  or  } \gamma > 1 \text{  and  }\pi_{\max} > 1,\\
\frac{a(\lambda)\tan{\left(\tan^{-1}{(b(\lambda)/a(\lambda))}+a(\lambda)x\right)}+\mu/\sigma^2-1/2}{\gamma-1},\\
\qquad \qquad \qquad \qquad \text{if  } \gamma > 1 \text{  and  } \pi_{\max} \in \left(\frac{1/\kappa-\sqrt{(1-1/\gamma)(2/\kappa-1)}}{2\left(\gamma/\kappa^2-(\gamma-1)(2/\kappa-1)\right)},\frac{1/\kappa+\sqrt{(1-1/\gamma)(2/\kappa-1)}}{2\left(\gamma/\kappa^2-(\gamma-1)(2/\kappa-1)\right)}\right),\\
\frac{a(\lambda)\coth{\left(\coth^{-1}{(b(\lambda)/a(\lambda))}-a(\lambda)x\right)}+\mu/\sigma^2-1/2}{\gamma-1},\\
\qquad \qquad \qquad \qquad \text{otherwise,}
\end{cases}
\end{displaymath}
with
\begin{eqnarray}
a(\lambda)&:=& \sqrt{\left|(\gamma-1)\frac{\mu^2}{\gamma \sigma^4}(1-(1-\kappa(1-\lambda))^2)-\left(\frac{1}{2}-\frac{\mu}{\sigma^2}\right)^2\right|},\nonumber\\
b(\lambda) &:=& \frac{1}{2}-\frac{\mu}{\sigma^2}+(\gamma-1)\pi_{\max}(1-\lambda).\nonumber
\end{eqnarray}
Then, for a sufficiently small $\lambda > 0$, the mapping $x \mapsto w(\lambda,x)$ is a local solution of 
\begin{align}
0 &= w'(x) + (1-\gamma)w(x)^2 +(\tfrac{2\mu}{\sigma^2}-1)w(x)\nonumber-\tfrac{\mu^2}{\gamma\sigma^4}(1-(1-\kappa(1-\lambda))^2),\nonumber\\
 w(0) &= (1-\lambda)\pi_{\max}.\label{w000}
\end{align}
Moreover, $x \mapsto w(\lambda,x)$ is increasing for $0<\pi_{\max} <1$ and decreasing for $\pi_{\max}>1$.
\end{mylemma}
\begin{proof}
The first part of the assertion is readily verified by taking derivatives. The second follows by inspection of the explicit formulae.
\end{proof}

 Next, we establish that the crucial constant $\lambda$, which determines both the no-trade region and the equivalent safe rate, is well-defined.
\begin{mylemma}\label{uniquelambda}
Let $0 < \pi_{\max}=\kappa \pi_* \not= 1$, define $w\left(\lambda,\cdot\right)$  as in Lemma~\ref{solw}, and set
\begin{equation}\label{deful}
l(\lambda) = \frac{(1-\lambda)\pi_{\max}}{1-(1-\lambda)\pi_{\max}} \quad \text{and} \quad  u=\frac{\pi_{\max}}{1-\pi_{\max}}.\end{equation}
Then, for a sufficiently small $\epsilon$, there exists a unique solution $\lambda$ of 
\begin{equation}\label{B2}
w\left(\lambda, \log{\left(u/l(\lambda)\right)}\right)=\frac{\pi_{\max}(1-\epsilon)}{(1-\pi_{\max})+\pi_{\max}(1-\epsilon)}=:w_+.
\end{equation}
As $\epsilon \downarrow 0$, it has the asymptotic expansion
\begin{displaymath}
\lambda =
\left(\frac{1}{\gamma}\frac{\kappa}{1-\kappa}\frac{(1-\pi_{\max})^2}{\pi_{\max}}\right)^{1/2}\epsilon^{1/2}+ \mathcal{O}(\epsilon).
\end{displaymath}
\end{mylemma}

\begin{proof} 
With minor modifications, this follows as in \cite[Lemma B.2]{gerhold.al.11}
\end{proof} 

Henceforth, $\lambda$ denotes the quantity from Lemma~\ref{uniquelambda}, and we omit the $\lambda$-dependence of $a=a(\lambda),b=b(\lambda)$, $l=l(\lambda)$, and $w(x)=w(\lambda,x)$. 

\begin{mycor}\label{wprimecor}
Let $0 < \pi_{\max} \not= 1$ and suppose $\epsilon$ is sufficiently small. Then, in all three cases,
\begin{align}
w'(0) &= \pi_{-}(1-\pi_{-})\label{wat0}\\
w'\left(\log{(u/l)}\right) &\leq  w_+ (1-w_+).\label{watlogul}
\end{align}
\end{mycor}

\begin{proof}
The assertions follow from the ODE for $w$ and its boundary conditions in Lemma~\ref{solw} and Lemma~\ref{uniquelambda}.
\end{proof}

\subsection{Shadow Prices and Verification}

A key idea for the proof of our verification theorem is to replace the original bid and ask prices by a single fictitious ``shadow price'' $\tilde{S}$ evolving within the bid-ask spread, which admits an optimal policy that is feasible (and hence also optimal) in the original market with transaction costs, too. An approach of this kind was first used in \cite{kallsen.muhlekarbe.10}, and has been utilized in the present setting modulo constraints by \cite{gerhold.al.11}.

\begin{mydef}\label{shadowpricedef}
A \emph{shadow price} is a process $\tilde S$ lying within the bid-ask spread $[(1-\epsilon)S, S]$, such that there exists a corresponding long-run optimal strategy $(\varphi^0,\varphi)$ of finite variation that satisfies the portfolio constraint \eqref{leverage constraint} and only entails buying (resp.\ selling) the risky asset when $\tilde{S}$ equals the buying (resp.\ selling) price.
\end{mydef}

The original constraints can be translated 
as follows:

\begin{myrek}\label{rek:shadowconstraint}
Let $\tilde{S}$ be a price process evolving within the bid ask spread $[(1-\epsilon)S,S]$. Then if a strategy $(\phi^0,\phi)$ satisfies the original portfolio constraint $\pi_t \leq \pi_{\max}$ on the risky weight computed with the ask price $S$,  it also satisfies the following constraint on the risky weight computed with $\tilde{S}$:
\begin{equation}\label{eq:shadowconstraint}
\tilde{\pi}_t \leq \tilde{\pi}_{\max}:=\begin{cases} \pi_{\max}, &\mbox{if } \pi_{\max} \leq 1,\\ \frac{\pi_{\max}(1-\epsilon)}{(1-\pi_{\max})+\pi_{\max}(1-\epsilon)}, & \mbox{otherwise}.\end{cases}
\end{equation}
\end{myrek}

  The construction successfully used in \cite{loewenstein.00,benedetti.al.11,herzegh.prokaj.11} suggests that the discounted shadow price can be constructed as the \emph{marginal rate of substitution} of risky for safe assets for the optimal investor, i.e., as the ratio of the partial derivatives of the value function with respect to the numbers of shares in the risky and safe asset, respectively:

$$\frac{\tilde{S}_t}{S^0_t}=\frac{\partial_{\varphi_t} V(t,X^0_t,X_t)}{\partial_{\varphi^0_t} V(t,X^0_t,X_t)}.$$
With the candidate value function derived in the above section, this leads to the candidate shadow price
\begin{equation}\label{eq:cand}
\tilde S_t = S_t\frac{w\left(\log{(X_t/X_t^0 l)}\right)}{X_t/X_t^0[1-w\left(\log{(X_t/(X_t^0 l))}\right)]} = S_t\frac{w\left(Y_t\right)}{l e^{Y_t}(1-w\left(Y_t\right))},
\end{equation}
where $e^{Y_t}= (X_t/X_t^0 l)$ is the ratio between the risky and safe positions at the ask price $S_t$, centered at the buying boundary $l= \frac{(1-\lambda)\pi_{\max}}{1-(1-\lambda)\pi_{\max}}$. In view of the above heuristics, the stock-cash ratio $X/X^0$ should remain within the no-trade region $[l,u]$; consequently, $Y$ should take values in $[0,\log(u/l)]$ (resp. $[\log{(u/l)},0]$, if $\pi_{\max} >1$). In the interior of this interval, the number of risky assets should remain constant, so that the dynamics of $Y = \log{\varphi/(l\varphi^0)}+\log{S/S^0}$ coincide with those of the Brownian motion $\log{S/S^0}$, which needs to be reflected at the boundaries to remain in $[0,\log(u/l)]$.

These heuristic arguments motivate to \emph{define} the process $Y$ as Brownian motion with instantaneous reflection at $0$ and $\log(u/l)$: 
\begin{equation}\label{dynamicsy}
d Y_t = (\mu-\sigma^2/2)d t + \sigma d W_t + d L_t - d U_t,
\end{equation}
where the local time processes $L$ and $U$ are adapted, continuous, non-decreasing (resp. non-increasing, if $\pi_{\max} >1$) and only increase (resp. decrease, if $\pi_{\max}>1$) on the sets $\{ Y_t = 0\}$ and $\{ Y_t = \log{(u/l)}\}$, respectively. The process $\tilde{S}$ can then be defined in accordance with \eqref{eq:cand}:

\begin{mylemma}\label{shatl}
Define 
\begin{equation}\label{y0}
y = \begin{cases}
0,& \text{if } l\xi^0 S_0^0 \geq \xi S_0,\\
\log{(u/l)},& \text{if } u\xi^0 S_0^0 \leq \xi S_0,\\
\log{\left[\xi S_0/(\xi^0 S_0^0 l)\right]},& \text{otherwise}.
\end{cases}
\end{equation}
and let $Y$ be defined as in~\eqref{dynamicsy}, starting at $Y_0 = y$. Then, $\tilde S = S \frac{w(Y)}{l e^{Y} (1-w(Y))}$, with $w$ as in Lemma~\ref{solw}, has the dynamics
\begin{displaymath}
\frac{d\tilde S (Y_t)}{\tilde S (Y_t)} = \left(\tilde{\mu}(Y_t)+r\right)d t+ \tilde{\sigma}(Y_t)d W_t + \left(1-\frac{w'}{(1-w)w}\left(\log{\left(\frac{u}{l}\right)}\right)\right)d U_t,
\end{displaymath}
where $\tilde \mu(\cdot)$ and $\tilde \sigma (\cdot)$ are given by
\begin{displaymath}
\tilde{\mu}(y) = \frac{\sigma^2w'(y)}{w(y)(1-w(y))}\left(\frac{w'(y)}{1-w(y)}-(1-\gamma)w(y)\right), \quad \tilde{\sigma}(y) = \frac{\sigma w'\left(y\right)}{w(y)(1-w(y))}.
\end{displaymath}
Moreover, the process $\tilde S$ takes values within the bid-ask spread $[(1-\epsilon)S,S]$.
\end{mylemma}
Note that the first two cases in~\eqref{y0} arise if the initial stock-cash ratio $\xi S_0/(\xi^0 S_0^0)$ lies outside of the interval $[l,u]$. Then, a jump from the initial position $(\varphi_{0^-}^0, \varphi_{0^-}) = (\xi^0,\xi)$ to the nearest boundary value of $[l,u]$ is required. This transfer necessitates the purchase resp. sale of the risky asset and hence the initial price $\tilde S _0$ is defined to match the buying resp.\ selling price of the risky asset.

\begin{proof}
The dynamics of $\tilde S$ result from It\^{o}'s formula, the dynamics of $Y$, and the identity
\begin{equation}\label{w2ableitung}
w''(y) = -2(1-\gamma)w'(y) w(y)- (2\mu/\sigma^2-1) w'(y),
\end{equation}
which is a direct consequence of Lemma~\ref{solw}. In addition, the boundary conditions for $w$ and $w'$ imply that
\begin{align*}
w''(0)-w'(0)+2w(0)w'(0) &= 2 w'(0) (\gamma w(0)-\frac{\mu}{\sigma^2})\\
&= 2\pi_{\max} (1-\lambda) (1-\pi_{\max} (1-\lambda))\frac{\mu}{\sigma^2}(\kappa (1-\lambda)-1)
\end{align*}
is negative (resp. positive, if $\pi_{\max}>1$). Thus, a comparison argument yields that the derivative of the function $\eta: y\mapsto \frac{w(y)}{l e^y (1-w(y))}$ is negative (resp. positive, if $\pi_{\max} >1$). Taking into account $\frac{w(0)}{l(1-w(0))}= 1$ and $\frac{w(\log{(u/l)})}{u(1-w(\log{(u/l)}))}=1-\epsilon$ completes the proof.
\end{proof}

Unlike in the absence of constraints, the dynamics of $\tilde S$ involve a singular part, such that the shadow market is no longer arbitrage-free. Indeed, whenever $\tilde{S}$ hits the lower bid price $(1-\varepsilon)S$ it is reflected upwards, such that one can make a riskless profit by buying risky assets and immediately selling them after the reflection has taken place. 

Due to the constraints \eqref{leverage constraint}, however, such arbitrage opportunities cannot be scaled arbitrarily. Consequently, there exists a discount factor that turns all admissible wealth processes into supermartingales (compare \cite{cvitanic.karatzas.93}):
 
\begin{mylemma}\label{discount}
For a fixed time horizon $T$, denote by $\tilde X_T^\psi$ the shadow payoff of an admissible strategy $(\psi^0,\psi)$ in the frictionless shadow market $(S^0,\tilde S)$, satisfying the constraint \eqref{leverage constraint} and hence also the shadow constraint $\tilde{\pi}^\psi \leq \tilde{\pi}_{\max}$ from Remark \ref{rek:shadowconstraint}.  Define the process $\tilde M$ by 
\begin{displaymath}
\tilde M_t :=
e^{-r t}\mathcal{E} \Big(-\int_0^\cdot \tfrac{\tilde \mu(Y_u)}{\tilde \sigma(Y_u)} d W_u\Big)_t e^{-\tilde{\pi}_{\max}\left(1-\frac{w'(\log(u/l))}{(1-w(\log(u/l)))w(\log(u/l))}\right)U_t},
\end{displaymath} 
with the local time process $U$ from~\eqref{dynamicsy}. Then $\tilde{M}$ is a discount factor:
\begin{displaymath}
\mathbb{E}[\tilde X_T^\psi \tilde M _T] \leq \tilde X _0^\psi.
\end{displaymath}
\end{mylemma}

\begin{proof}
First, notice that $\tilde \mu, \tilde \sigma$ and $w$ are functions of $Y$, but the argument is omitted throughout to ease notation. Inserting the dynamics of $\tilde S$ yields
\begin{displaymath}
\tilde X _T^\psi = \tilde X _0^\psi \mathcal{E}\left(\int_0^\cdot \left(r+ \tilde \pi_t^{\psi} \tilde \mu\right) d t + \int_0^\cdot \tilde \pi_t^{\psi} \tilde \sigma d W_t\right)_T e^{\int_0^T\tilde \pi_t^{\psi} \left(1-\frac{w'}{(1-w)w}\right)dU_t},
\end{displaymath}
where $\tilde \pi^\psi$ denotes the risky weight in terms of $\tilde S$. In view of Remark \ref{rek:shadowconstraint}, we have $\tilde \pi^{\psi} \leq \tilde{\pi}_{\max}$. Furthermore, the identity~\eqref{watlogul} implies $(1-\frac{w'}{(1-w)w})\geq 0$ (resp. $\leq 0$, if $\pi_{\max}>1$) on the set $\{y=\log(u/l)\}$ where $U$ increases (resp. decreases, if $\pi_{\max}>1$). Hence,
\begin{align*}
\mathbb{E}[\tilde X_T^{\psi} \tilde M_T] = & \mathbb{E} \Bigl[\tilde X _0^\psi\mathcal{E}\Big(\int_0^\cdot (r+ \tilde \pi^{\psi}_t \tilde \mu) d t + \int_0^\cdot \tilde \pi^{\psi}_t \tilde \sigma d W_t\Big)_T\\
&  \qquad \times \mathcal{E}\Big(\int_0^\cdot - r d t - \int_0^\cdot \frac{\tilde \mu}{\tilde \sigma} d W_t\Big)_T e^{\int_0^T (\tilde{\pi}^\psi_t-\tilde{\pi}_{\max}) (1-\frac{w'}{(1-w)w})dU_t}\Bigr]\\
\leq& \mathbb{E}\left[\tilde X _0^\psi\mathcal{E}\left(\int_0^\cdot \left(\tilde \pi^{\psi} \tilde \sigma -\frac{\tilde \mu}{\tilde \sigma}\right)d W_t\right)_T\right] \leq \tilde X _0^\psi,\nonumber
\end{align*}
where we have used for the last inequality that the positive local martingale $\mathcal{E}(\int_0^\cdot (\tilde \pi^{\psi} \tilde \sigma -\frac{\tilde \mu}{\tilde \sigma})d W_t)$ is a supermartingale.
\end{proof}

If the candidate process $\tilde S$ from Lemma \ref{shatl} is indeed a shadow price, then the optimal numbers of safe and risky assets should be the same as in the original market with transaction costs. By definition of $\tilde{S}$, the candidate optimal strategy derived heuristically in Section~\ref{heuristic} therefore leads to the following candidate for the long-run optimal risky weight in the shadow market:
\begin{displaymath}
\tilde \pi (Y_t) = \frac{\varphi_t \tilde S _t }{\varphi_t^0 S_t^0+ \varphi_t \tilde S _t}
 = \frac{\varphi_t S _t\frac{w(Y_t)}{l e^{Y_t}(1-w(Y_t))} }{\varphi_t^0 S_t^0+ \varphi_t S _t\frac{w(Y_t)}{l e^{Y_t}(1-w(Y_t))}}
   =\frac{\frac{w(Y_t)}{1-w(Y_t)} }{1+ \frac{w(Y_t)}{1-w(Y_t)}}
   = w(Y_t).
\end{displaymath}

To show that this risky weight is indeed long-run optimal for $\tilde{S}$, we first establish the following finite-horizon bounds in analogy to the frictionless case \cite{guasoni.robertson.12}:

\begin{mylemma}\label{LemmaC2}
For a fixed time horizon $T>0$, let $\beta = \frac{\mu^2}{2\gamma\sigma^2}(1-(1-\kappa(1-\lambda))^2)$ and let the function $w$ be defined as in Lemma~\ref{solw}. Then, for the the shadow payoff $\tilde X_T$  corresponding to the policy $\tilde \pi(Y) = w(Y)$ and the shadow discount factor $\tilde M_T$ introduced in Lemma~\ref{discount}, the following bounds hold true: 
\begin{align}
\mathbb{E} [\tilde X _T^{1-\gamma}] &= \tilde X_0 ^{1-\gamma} e^{(1-\gamma)(r+\beta)T}\hat{\mathbb{E}}[e^{(1-\gamma)\left(\tilde q (Y_T)- \tilde q (Y_0) \right)}],\label{bound1}\\
\mathbb{E}[\tilde M _T^{1-\frac{1}{\gamma}}]^{\gamma} &= e^{(1-\gamma)(r+\beta)T}\hat{\mathbb{E}}[e^{(\frac{1}{\gamma}-1)\left(\tilde q (Y_T)- \tilde q (Y_0) \right)}]^{\gamma}, \label{bound2}
\end{align}
where $\tilde q (y) := \int_0^y (\frac{w'(z)}{1-w(z)}-w(z)) dz$ and $\hat{\mathbb{E}} \left[\cdot\right]$ denotes the expectation with respect to the \emph{myopic probability} $\hat{\mathbb{P}}$, defined by
\begin{displaymath}
\frac{d \hat{\mathbb{P}}}{d \mathbb{P}} = \exp\left(\int_0^T \left(-\frac{\tilde \mu(Y_t)}{\tilde \sigma(Y_t)} + \tilde \sigma(Y_t) \tilde \pi(Y_t)\right) d W_t - \frac{1}{2}\int_{0}^T\left(-\frac{\tilde \mu(Y_t)}{\tilde \sigma(Y_t)}+ \tilde \sigma(Y_t) \tilde \pi(Y_t)\right)^2 d t\right).
\end{displaymath}
\end{mylemma}

\begin{proof}
To ease the notation, we again omit the argument $Y$ of the functions $\tilde \mu,\tilde \sigma,\tilde \pi$ and $w$. 
To obtain~\eqref{bound1}, notice that
\begin{displaymath}
\tilde X _T^{1-\gamma} = \tilde X _0^{1-\gamma} e^{(1-\gamma)\int_0^T (r+ \tilde \mu w -\frac{\tilde \sigma^2}{2} w^2)d t+ (1-\gamma)\int_0^T \tilde \sigma w d W_t+ \int_0^T (1-\gamma)(w-\frac{w'}{1-w}) U_t}.
\end{displaymath}
Hence,
\begin{align*}
\tilde X _T^{1-\gamma} =& \tilde X _0^{1-\gamma}\frac{d \hat{\mathbb{P}}}{d \mathbb{P}}  e^{\int_0^T \big[(1-\gamma)(r+ \tilde \mu w -\frac{\tilde \sigma^2}{2}w^2) + \frac{1}{2} (-\frac{\tilde \mu}{\tilde \sigma}+ \tilde \sigma w)^2\big]d t}\nonumber\\
 & \quad  \times e^{\int_0^T \big[(1-\gamma)\tilde \sigma w-(-\frac{\tilde \mu}{\tilde \sigma}+\tilde \sigma w)\big] d W_t+ \int_0^T(1-\gamma)(w-\frac{w'}{1-w})dU_t}.\nonumber
\end{align*}
Inserting the definitions of $\tilde \mu$ and $\tilde \sigma$, the second integrand simplifies to $(1-\gamma)\sigma(\frac{w'}{1-w}-w)$. Similarly, the first integrand reduces to $(1-\gamma) (r+\frac{\sigma^2}{2}(\frac{w'}{1-w})^2-(1-\gamma)\sigma^2w (\frac{w'}{1-w})+(1-\gamma)\frac{\sigma^2}{2}w^2)$. In summary:
\begin{eqnarray}\label{wpart}
\tilde X _T^{1-\gamma} &=& \tilde X _0^{1-\gamma}\frac{d \hat{\mathbb{P}}}{d \mathbb{P}}  e^{(1-\gamma) \int_0^T \left(r+\frac{\sigma^2}{2}\left(\frac{w'}{1-w}\right)^2-(1-\gamma)\sigma^2w\left(\frac{w'}{1-w}\right)+(1-\gamma)\frac{\sigma^2}{2}w^2\right)d t}\nonumber\\
 & & \times e^{(1-\gamma) \left[\int_0^T \left(\frac{w'}{1-w}-w\right)\sigma d W_t - \int_0^T\left(\frac{w'}{1-w}-w\right)dU_t\right]}.
\end{eqnarray}
It\^{o}'s formula and the boundary conditions for $w$ imply that
\begin{eqnarray}
\tilde q (Y_T) -\tilde q (Y_0) &=& \int_0^T \left(\frac{w'}{1-w}-w\right) \sigma d W_t - \int_0^T\left(\frac{w'}{1-w}-w\right)dU_t\nonumber\\
& & +\int_0^T \left(\mu-\frac{\sigma ^2}{2}\right)\left(\frac{w'}{1-w}-w\right)+ \frac{\sigma^2}{2}\left(\frac{w''(1-w)+w'^2}{(1-w)^2}-w'\right)d t.\label{itow}
\end{eqnarray}
Substituting the second derivative $w''$ according to the equation~\eqref{w2ableitung} and using this identity to replace the stochastic integral in~\eqref{wpart} yields 
\begin{displaymath}
\tilde X _T^{1-\gamma} =\tilde X _0^{1-\gamma} \frac{d \hat{\mathbb{P}}}{d \mathbb{P}} e^{(1-\gamma) \int_0^T \left(r+\frac{\sigma^2}{2}w'+(1-\gamma)\frac{\sigma^2}{2}w^2+\left(\mu-\frac{\sigma^2}{2}\right)w\right) d t} e^{(1-\gamma)\left(\tilde q(Y_T)-\tilde q(Y_0)\right)}.
\end{displaymath}
Thus, the first bound results from taking expectations on both sides and using the ODE for $w$ (cf.\ Lemma~\ref{solw}). 

Similarly, plugging in the definitions of $\tilde \mu$ and $\tilde \sigma$, the (shadow) discount factor $\tilde M_T$ and the myopic probability measure $\hat{\mathbb{P}}$ satisfy
\begin{eqnarray}
\tilde M _T^{1-\frac{1}{\gamma}} &=& e^{\frac{1-\gamma}{\gamma}\int_0^T \frac{\tilde \mu}{\tilde \sigma} d W_t + \frac{1-\gamma}{\gamma}\int_0^T \left(r+\frac{\tilde \mu ^2}{2 \tilde \sigma ^2}\right) d t +\frac{1-\gamma}{\gamma}\int_0^T\left(w-\frac{w'}{1-w}\right)dU_t}\nonumber\\
 &=& \frac{d \hat{\mathbb{P}}}{d \mathbb{P}} e^{\frac{1-\gamma}{\gamma}\int_0^T \left( \frac{\tilde \mu}{\tilde \sigma}-\frac{\gamma}{1-\gamma}\left(-\frac{\tilde \mu}{\tilde \sigma}+\tilde \sigma w\right)\right) d W_t + \frac{1-\gamma}{\gamma}\int_0^T\left(w-\frac{w'}{1-w}\right)dU_t}\nonumber\\
 & &\times e^{\frac{1-\gamma}{\gamma}\int_0^T \left(r+\frac{\tilde \mu ^2}{2 \tilde \sigma ^2}+\frac{\gamma}{2(1-\gamma)}\left(-\frac{\tilde \mu}{\tilde \sigma}+\tilde \sigma w\right)^2\right) d t}\nonumber\\
 &=& \frac{d \hat{\mathbb{P}}}{d \mathbb{P}} e^{\frac{1-\gamma}{\gamma}\int_0^T  \left(\frac{w'}{1-w}-w\right) \sigma d W_t - \frac{1-\gamma}{\gamma}\int_0^T\left(\frac{w'}{1-w}-w\right)dU_t}\nonumber\\
 & &\times e^{\frac{1-\gamma}{\gamma}\int_0^T  \left(r+\frac{\sigma^2}{2}\left(\frac{w'}{1-w}\right)^2-(1-\gamma)\sigma^2w\left(\frac{w'}{1-w}\right)+(1-\gamma)\frac{\sigma^2}{2}w^2\right) d t}.\nonumber
\end{eqnarray}
Substituting again the stochastic integral with equation~\eqref{itow} and~\eqref{w2ableitung} we obtain
\begin{displaymath}
\tilde M _T^{1-\frac{1}{\gamma}} = \frac{d \hat{\mathbb{P}}}{d \mathbb{P}}  e^{\frac{1-\gamma}{\gamma} \int_0^T \left(r+\frac{\sigma^2}{2}w'+(1-\gamma)\frac{\sigma^2}{2}w^2+\left(\mu-\frac{\sigma^2}{2}\right)w\right) d t} e^{\frac{1-\gamma}{\gamma}\left(\tilde q(Y_T)-\tilde q(Y_0)\right)}.
\end{displaymath}
Finally, the second bound results from taking the expectations on both sides, raising it to the power of $\gamma$, and using the ODE for $w$ from Lemma~\ref{solw}.
\end{proof}

With the finite-horizon bounds at hand, it is now straightforward to establish the long-run optimality of the policy $\tilde \pi (Y) = w(Y)$ in the shadow market. Note that this strategy does not take advantage of the available arbitrage opportunities: It already attains the portfolio constraint whenever the singular component of $\tilde{S}$ acts, hence the constraints prevent it from further increasing its risky position to profit from this.\footnote{A similar example of a market with arbitrage opportunities for which the optimal policy does not exploit these can be found in \cite[Example 5.1]{goll.kallsen.03}. Whereas the latter is somewhat pathological, this situation arises naturally in the present context when looking for an equivalent frictionless shadow market.}
\begin{mylemma}\label{LemmaC3}
Let $0 < \pi_{\max} \not =1$ and let $w$ be defined as in Lemma~\ref{solw}. Then, the policy $\tilde \pi (Y) = w(Y)$ is long-run optimal in the shadow market $(S^0,\tilde S)$ with equivalent safe rate $r+ \beta$, where $\beta$ is specified in Lemma~\ref{LemmaC2}. For $t \geq 0$, the corresponding wealth process and the numbers of safe and risky assets are given by
\begin{gather*}
\tilde X_t = (\xi^0S_0^0+ \xi\tilde S_0) \mathcal{E}\Bigl(\int_0^\cdot (r+ w(Y_s)\tilde \mu(Y_s)) d s+ \int_0^\cdot w\left(Y_s\right) \tilde \sigma \left(Y_s\right) d W_s +\int_0^\cdot \left(w\left(Y_s\right)-\tfrac{w'\left(Y_s\right)}{1-\left(w\left(Y_s\right)\right)}\right) d U_s\Bigr)_t,\\
\varphi_{0^-} = \xi, \quad \varphi_t = w \left(Y_t\right) \tilde X _t/\tilde S _t, \quad \varphi_{0^-}^0 = \xi^0, \quad \varphi_t^0 = \left(1-w \left(Y_t\right)\right) \tilde X _t/S _t^0.
\end{gather*}
\end{mylemma}

\begin{proof}
The formulas for the wealth process and the number of safe and risky units follow directly from the respective definitions. By definition of $\tilde{S}$ and the formulas for $\varphi^0,\varphi$ the corresponding risky weight in terms of the ask price is given by $\pi_t=\frac{le^{Y_t}}{1+le^{Y_t}}$. Since $Y$ is reflected to remain between $0$ and $\log(u/l)$, it takes values between $\frac{l}{1+l}=(1-\lambda)\pi_{\max}$ and $\frac{u}{1+u}=\pi_{\max}$, such that the constraint \eqref{leverage constraint} is satisfied.

 To verify the optimality of this policy, we use the standard duality bound for power utility (cf.\ \cite[Lemma 5]{guasoni.robertson.12}, which is applicable by Lemma \ref{discount}), valid for the shadow payoff $\tilde X _T^\psi$ of any admissible strategy $(\psi^0,\psi)$ in the shadow market:
\begin{displaymath}
\mathbb{E}[(\tilde X_T^{\psi})^{1-\gamma}]^{\frac{1}{1-\gamma}} \leq \mathbb{E}[(\tilde M_T)^{\frac{\gamma-1}{\gamma}}]^{\frac{\gamma}{1-\gamma}}.
\end{displaymath}
This inequality and the second bound~\eqref{bound2} in Lemma~\ref{LemmaC2} yield an upper bound for the equivalent safe rate:
\begin{displaymath}
\liminf_{T\rightarrow \infty} \frac{1}{(1-\gamma)T}\log{\mathbb{E}[(\tilde X_T^{\psi})^{1-\gamma}]} \leq \liminf_{T\rightarrow \infty} \frac{\gamma}{(1-\gamma)T}\log{\mathbb{E}[(\tilde M_T ^{\frac{\gamma-1}{\gamma}})]} = r+ \beta.\nonumber
\end{displaymath}
Here, we used in the last step that $\tilde q$ is bounded on the compact support of its argument $Y$. Similarly, the first bound~\eqref{bound1} in Lemma~\ref{LemmaC2} implies that the upper bound $r+ \beta$ is attained by the policy $\tilde \pi$, which corresponds to the strategy $(\varphi^0,\varphi)$.
\end{proof}

To conclude that $\tilde S$ is indeed a shadow price, it remains to check that the optimal strategy $(\varphi^0,\varphi)$ only acts on the sets $\{Y_t = 0\}$ and $\{Y_t = \log{\left(u/l\right)}\}$. 
\begin{mylemma}\label{LemmaC4}
Let $0 < \pi_{\max} \not= 1$. Then, the number of shares $\varphi = w(Y)\tilde X/\tilde S$ for the optimal policy from Lemma~\ref{LemmaC3} has the dynamics
\begin{displaymath}
\frac{d \varphi_t}{\varphi_t} = (1-\pi_{\max}(1-\lambda))d L_t - \frac{1-\pi_{\max}}{1-\pi_{\max}+\pi_{\max}(1-\epsilon)}dU_t.
\end{displaymath}
Thus, $\varphi_t$ increases only when $\{Y_t = 0\}$, that is, when $\tilde S_t$ equals the ask price, and decreases only when $\{Y_t = \log{\left(u/l\right)}\}$, that is, when $\tilde S_t$ equals the bid price. In particular, it is of finite variation.
\end{mylemma}
\begin{proof}
It\^{o}'s formula, the boundary conditions for $w$, and Equation~\eqref{w2ableitung} imply
\begin{displaymath}
d w(Y_t) = -(1-\gamma)\sigma^2w'(Y_t)w(Y_t)d t + \sigma w'(Y_t)d W_t + w'(Y_t) (d L_t-d U_t).
\end{displaymath}
Integrating $\varphi = w(Y) \tilde X/\tilde S$ by parts twice, using the dynamics of $w(Y),\tilde X$ and $\tilde S$, and simplifying yields
\begin{displaymath}
d \varphi_t/\varphi_t = w'(Y_t)/w(Y_t) d L_t - (1-w(Y_t))d U_t.
\end{displaymath}
Since $L_t$ and $U_t$ only increase (resp. decrease, if $\pi_{\max} >1$) on $\{Y_t = 0\}$ and $\{Y_t = \log{\left(u/l\right)}\}$, respectively, the assertion then follows from the boundary conditions for $w$ and identity~\eqref{wat0}.
\end{proof}

Since the shadow price takes values in the bid-ask spread, it allows the investor to trade at more favorable prices than in the original market with transaction costs. As the optimal strategy $(\varphi^0,\varphi)$ only entails buying (resp.\ selling) the risky asset when $\tilde S$ coincides with the ask (resp.\ bid) price, it is also feasible with transaction costs, and therefore optimal with the same growth rate. 

\begin{mylemma}\label{LemmaC5}
For sufficiently small $\epsilon$, the policy $(\varphi^0,\varphi)$ from Lemma \ref{LemmaC3}  is also long-run optimal in the original market with transaction costs, with the same equivalent safe rate $r+\beta$.
\end{mylemma}

\begin{proof}
This follows along the lines of \cite[Proposition C.5]{gerhold.al.11}.
\end{proof}

Let us summarize our results so far.

\begin{mythm}\label{TheoremC6}
For a small spread $\epsilon > 0$, and a binding portfolio constraint $0<\pi_{\max} \not= 1$, the process $\tilde S$ in Lemma~\ref{shatl} is a shadow price in the sense of Definition~\ref{shadowpricedef}. A long-run optimal policy with equivalent safe rate $r+\beta$, both for the shadow market and for the original market with transaction costs, is to keep the risky weight $\pi$ (in terms of the ask price $S$) within the no-trade region
\begin{displaymath}
\left[\pi_-,\pi_+\right] = [(1-\lambda)\pi_{\max},\pi_{\max}].
\end{displaymath}
\end{mythm}

\subsection{Trading Volume}\label{shwe}
The formulas for share and wealth turnover follow from \cite[Lemma D.2]{gerhold.al.11} along the lines of \cite[Corollary D.3]{gerhold.al.11}. Plugging the asymptotics for the gap $\lambda$ into these explicit formulas yields the corresponding expansions for share and wealth turnover.

\subsection{Approximate Optimality for Finite Horizons}
Finally we show that, like in the absence of constraints, the stationary long-run optimal policy is also approximately optimal for any finite horizon $T>0$.

\begin{proof}[Proof of Theorem~\ref{finitehorizone}]
For a fixed time horizon $T$ let $(\psi^0,\psi)$ be any admissible strategy. The liquidation value $\Xi_T^\psi$ of this strategy is smaller than the corresponding shadow payoff $\tilde X _T^\psi$, i.e., $\Xi_T^\psi \leq \tilde X_T^\psi = \psi_0^0 + \psi_0\tilde S_0+ \int_0^T\psi_s d \tilde S _s$. Thus, the standard duality bound for power utility as in \cite[Lemma 5]{guasoni.robertson.12} and the second bound~\eqref{bound2} in Lemma~\ref{LemmaC2} imply that
\begin{equation}\label{C13}
\frac{1}{(1-\gamma)T}\log{\mathbb{E}[(\Xi_T^\psi)^{1-\gamma}]}
\leq r+ \beta + \frac{1}{T}\log{(\psi_{0^-}^0+\psi_{0^-}S_0)} +\frac{\gamma}{(1-\gamma)T}\log{\hat{\mathbb{E}}[e^{{(\frac{1}{\gamma}-1)}(\tilde q (Y_T)- \tilde q (Y_0))}]}.
\end{equation} 
For the long-run optimal strategy $(\varphi^0,\varphi)$ we have $\Xi_T^\varphi \geq (1-\pi_{\max}\epsilon/(1-\epsilon))\tilde X_T$. Hence, the first bound~\eqref{bound1} in Lemma~\ref{LemmaC2} yields
\begin{eqnarray}\label{C14}
\frac{1}{(1-\gamma)T}\log{\mathbb{E}[(\Xi_T^\varphi)^{1-\gamma}]} &\geq& 
r+ \beta + \frac{1}{T}\log{(\varphi_{0^-}^0+\varphi_{0^-}\tilde S_0)}\\
 & &  + \frac{1}{(1-\gamma)T}\log{\hat{\mathbb{E}}[e^{\left(1-\gamma\right)\left(\tilde q (Y_T)- \tilde q (Y_0)\right)}]} + \frac{1}{T}\log{\left(1-\frac{\epsilon}{1-\epsilon}\pi_{\max}\right)}.\nonumber
\end{eqnarray}
For $\epsilon \downarrow 0$, we have $\log{(1-\pi_{\max} \epsilon/(1-\epsilon))} = \mathcal{O}(\epsilon)$ and $\log{(\varphi_{0^-}^0+\varphi_{0^-}\tilde S_0)} \geq  \log{(\varphi_{0^-}^0+\varphi_{0^-} S_0)} + \mathcal{O}(\epsilon)$ because $\tilde S \in [(1-\epsilon)S,S]$. Moreover, since $\lambda = \mathcal{O}(\epsilon^{1/2})$ and $\log{(u/l)} = \mathcal{O}(\epsilon^{1/2})$, the identities ~\eqref{wat0},~\eqref{w2ableitung} and Taylor expansion yield
\begin{equation*}
\frac{w'(z)}{1-w(z)}-w(z) =
\frac{\kappa^2\mu^2}{\gamma \sigma^2}\left(1-\frac{1}{\kappa}\right)\epsilon^{1/2} + \mathcal{O}\left(\epsilon\right).
\end{equation*}
Therefore, $\tilde q(y) = \int_0^y (\frac{w'(z)}{1-w(z)}-w(z)) dz$ is also of order $\mathcal{O}\left(\epsilon\right)$ for $y \leq |\log(u/l)|$, which completes the proof.
\end{proof}
\bibliographystyle{abbrv}
\bibliography{trans_constr}

\end{document}